\documentclass[11pt, dvipsnames]{article}
\usepackage[T1]{fontenc}

\usepackage[margin=1in]{geometry}
\pdfoutput =1
\usepackage{setspace}
\usepackage{amsthm}
\usepackage{cite}
\usepackage{amsmath,amssymb,amsfonts}
\usepackage{algorithmic}
\usepackage{graphicx}
\usepackage{macrosRouting}
\usepackage{macros}
\usepackage{algorithm,algorithmic}

\usepackage{appendix}
\title{Adaptive Incentive Design with Learning Agents
}

\author{Chinmay Maheshwari\thanks{CM, KK and SS are with EECS,UC Berkeley, CA, United States.} , Kshitij Kulkarni$^\ast$, Manxi Wu\thanks{MW is with ORIE, Cornell University, NY, United States.
}, and Shankar Sastry$^\ast$
}

\begin{document}
\maketitle

\begin{abstract}
   {We propose an adaptive incentive mechanism that learns the optimal incentives in environments where players continuously update their strategies.} Our mechanism updates incentives based on each player's externality, defined as the difference between the player's marginal cost and the operator's marginal cost at each time step.  
The proposed mechanism updates the incentives on a slower timescale compared to the players' learning dynamics, resulting in a two-timescale coupled dynamical system. Notably, this mechanism is agnostic to the specific learning dynamics used by players to update their strategies.  
We show that any fixed point of this adaptive incentive mechanism corresponds to the optimal incentive mechanism, ensuring that the Nash equilibrium coincides with the socially optimal strategy.  
 Additionally, {we provide sufficient conditions under which the adaptive mechanism converges to a fixed point.} Our results apply to both atomic and non-atomic games.
To demonstrate the effectiveness of our proposed mechanism, we verify the convergence conditions in two practically relevant classes of games: atomic aggregative games and non-atomic routing games.
\end{abstract}
\newpage 
\label{sec:introduction}
\section{Introduction}
Incentive mechanisms play a crucial role in many societal systems, where outcomes are governed by the interactions of a large number of self-interested users (or algorithms acting on their behalf).  
The outcome of such strategic interactions, characterized by the Nash equilibrium, is often suboptimal because individual players typically do not account for the \textit{externality} of their actions (i.e., how their actions affect the costs of others) when minimizing their own costs.  
An important way to address this suboptimality is to provide players with incentives that align their individual goal of cost minimization with the goal of minimizing the total cost of the societal system \cite{levin1985taxation, bacsar1984affine}.  
{However, this problem becomes more challenging as the system operator often needs to account for the learning behavior of players, who repeatedly update their strategies in response to the incentive mechanism, especially when the physical system experiences a random shock and players are learning to reach a new equilibrium \cite{barrera2014dynamic,como2021distributed,maheshwari2021dynamic}.  
}

To address this challenge, we propose an adaptive incentive mechanism that adjusts incentives based on the strategies of players, who repeatedly update their strategies as part of a learning process.  
This results in a \textit{coupled} dynamical system that comprises both incentive and strategy updates.

{Our proposed incentive mechanism has four key features. Firstly, our framework applies to both atomic and non-atomic games.} Secondly, the incentive update incorporates the \textit{externality} generated by the players' current strategies, quantified as the difference between their own marginal cost and the marginal cost for the entire system. 
{Thirdly, the incentive mechanism is agnostic to the strategy update dynamics used by players and requires only oracle access to either the gradient (in atomic games) or the value (in non-atomic games) of the cost function, given the current strategy, to evaluate the externality.}
 Finally, the incentive update occurs on a \textit{slower timescale} compared to the players' strategy updates.  
This slower evolution of incentives is a desirable characteristic because frequent incentive updates often hinder players' participation.

We prove that any fixed point of the coupled incentive and strategy updates leads to a socially optimal outcome. Specifically, at any fixed point, the incentive provided to each player equals the externality of the equilibrium strategy, ensuring that the resulting Nash equilibrium is socially optimal (Proposition \ref{prop: Alignment}).  
Additionally, we establish sufficient conditions on the underlying game that ensure the fixed point -- coinciding with the socially optimal incentive mechanism -- is unique (Proposition \ref{prop: Alignment}).

{We characterize sufficient conditions for (both local and global) convergence of the coupled dynamical system  to the fixed points (Proposition \ref{thm: ConvergenceFin}).} Since the convergent strategy profile and incentive mechanism correspond to a socially optimal outcome, these sufficient conditions ensure that the coupled dynamical system  induce a socially optimal outcome in the long run. Our analysis builds on the theory of two-timescale dynamical systems \cite{borkar1997stochastic}.  
Due to the timescale separation between the strategy and incentive updates, we can decouple the convergence of the strategy update from that of the incentive update. 
First, the convergence of the strategy update, which evolves on the faster timescale, is analyzed by treating the incentive mechanism, which evolves on the slower timescale, as static.  
In our work, we offload this analysis to the extensive literature on learning in games (e.g., \cite{mertikopoulos2024unified, leslie2006generalised, benaim2005stochastic, sandholm2010population, swenson2018best}).  
Second, the convergence of the incentive mechanism update is examined through the corresponding continuous-time dynamical system, evaluated at the fixed point of the strategy update (i.e., the Nash equilibrium).  

To demonstrate the usefulness of the adaptive incentive mechanism, we apply it to two practically relevant classes of games: \textit{(i)} atomic aggregative games and \textit{(ii)} non-atomic routing games.  
In atomic aggregative games, each player's cost function depends on their own strategy as well as the aggregate strategies of their opponents.  
This aggregation is performed through a linear combination of neighboring players' strategies, with weights characterized by a \textit{network matrix}.  
Our proposed incentive mechanism enables the system operator to adaptively adjust incentives based on each player's externality on their neighbors while players learn their equilibrium strategies.  
When applied to this setting, our results provide sufficient conditions on the network matrix to ensure {global} convergence to a socially optimal outcome.

Furthermore, in non-atomic routing games, players (travelers) make routing decisions in a congested network with multiple origin-destination pairs.  
The system operator imposes incentives in the form of toll prices on network edges.  
Our proposed incentive mechanism is adaptively updated based solely on the observed edge flows and the gradient of the edge latency functions.  
Players can follow various strategy update rules that lead to the equilibrium of the routing game.  
We show that the adaptive incentive mechanism {locally converges} to the toll prices that minimize total congestion.  

The article is organized as follows: In Sec. \ref{sec: Model}, we describe the setup for both atomic and non-atomic games and introduce the joint strategy and incentive update framework. Sec. \ref{sec: Results} presents our results on the fixed points being socially optimal, and sufficient conditions for local and global convergence in general games. In Sec. \ref{sec: Applications}, we apply these convergence results to atomic aggregative games (Sec. \ref{sec:aggregate}) and non-atomic routing games (Sec. \ref{sec:routing}). Finally, we conclude in Sec. \ref{sec: Conclusion}. 

\subsection{Related Works} 
\textbf{Two-timescale Learning Dynamics:} Learning dynamics in which incentives are updated on a slower timescale than players' strategies have been studied in \cite{mojica2022stackelberg, chen2023high, poveda2017class, ochoa2022high, liu2021inducing, li2023inducing, alpcan2009nash, alpcan2009control}. Specifically, \cite{mojica2022stackelberg} examines Stackelberg games with a single leader and a population of followers, where the leader employs gradient-based updates while the followers adjust their strategies using replicator dynamics. Moreover, \cite{chen2023high, poveda2017class} focus on incentive design in affine congestion games, where incentives are updated using a distributed version of gradient descent. Similarly, \cite{ochoa2022high} studies incentive design for traffic control on a single highway through gradient-based incentive updates.
Additionally, \cite{liu2021inducing, li2023inducing} propose a two-timescale discrete-time learning dynamic in which players update their strategies using mirror descent, while the system operator adjusts the incentive parameter via a gradient-based method.
{Furthermore, \cite{alpcan2009nash, alpcan2009control} study the  convergence of gradient-based incentive updates when the system operator has access to the gradient of the equilibrium strategy with respect to the incentive.}

{All of these works adopt gradient-based incentive updates. In such approaches, ensuring that the fixed point is socially optimal relies on the assumption that the equilibrium social cost is a convex function of the incentive parameter \cite{mojica2022stackelberg, chen2023high, poveda2017class, ochoa2022high, liu2021inducing, li2023inducing} or that the gradient of the equilibrium strategy with respect to the incentive is non-singular \cite{alpcan2009control, alpcan2009nash}. However, these assumptions are restrictive and often do not hold, even in simple games. In Appendix \ref{sec: AppendixCounterexample}, we provide a counterexample—a two-link routing game—in which both the convexity and non-singular gradient assumptions fail to hold.}

\textbf{Single-timescale Learning Dynamics:} The problem of \textit{steering} non-cooperative players toward a desired Nash equilibrium using an incentive update that operates on the same timescale as strategy updates has been studied in \cite{shakarami2022dynamic, shakarami2022steering, zhang2023steering, ratliff2020adaptive}. Specifically, \cite{shakarami2022dynamic} examines such updates in the setting of quadratic aggregative games. In \cite{shakarami2022steering}, 
the authors consider a scenario where players' costs depend only on their own actions and a price signal provided by an operator. 
In \cite{zhang2023steering}, the authors address the problem of guiding no-regret learning players toward an optimal equilibrium; however, their approach requires solving an optimization problem at each time step to compute the incentive mechanism.
The work in \cite{ratliff2020adaptive} explores incentive design while simultaneously learning players' cost functions. The authors assume that both cost functions and incentive policies are linearly parameterized, with incentive updates relying on knowledge of the players’ strategy update rules rather than solely on their current strategies, as in our setting.

\textbf{Learning in Stackelberg Games:} Our work is also related to the literature on learning in Stackelberg games, where the planner often has limited information about the interactions between players and must design an optimal mechanism by dynamically incorporating feedback from players' responses (see, e.g., \cite{blum2014learning, letchford2009learning, peng2019learning, bai2021sample, cui2024learning, jalota2022online, maheshwari2023follower}). This line of research typically imposes structural assumptions on the game among followers, such as a finite action space or linearly parameterized utility functions \cite{blum2014learning, letchford2009learning, peng2019learning, bai2021sample, cui2024learning, jalota2022online}. Alternatively, some works, such as \cite{maheshwari2023follower}, focus on ensuring convergence only to a locally optimal solution.

{Compared to the preceding three lines of research, we introduce a novel externality-based adaptive incentive design that applies to both atomic and nonatomic games, accommodates continuous action spaces, and allows for nonlinear utility functions. Unlike gradient-based incentive updates, externality-based updates ensure that any fixed point of the dynamics is socially optimal without requiring the equilibrium social cost function to be convex in the incentive vector {or the gradient of equilibrium strategy with respect to the incentive to be non-singular}. Furthermore, our incentive update is agnostic to the players' learning dynamics and relies only on oracle access to zeroth-order or first-order information about players' costs given their current strategies.}

\subsection*{Notations}
Given a function \(f:\R^{n}\rightarrow \R\), we use \(\nabla_{x_i} f(x)\) to denote the partial derivative of $f$ with respect to $x_i$ for any \(i\in\{1,2,...,n\},\) and \(\nabla f(x)\) to denote the gradient of the function. 
For any set \(A\), we use \(\textsf{conv}(A)\) to denote its convex hull.
For any set \(X\subseteq\mathbb{R}^n\), a function \(f: X\rightarrow\mathbb{R}\) is \textit{Lipschitz} if there exists a positive scalar \(L\) such that \(\|f(x)-f(x')\|\leq L\|x-x'\|,\) for every \(x,x'\in X\). For any vector $x \in X$ and any positive scalar $r>0$, the set $\mathcal{B}_{r}(x) = \{x' \in X | \|x'-x\|<r\}$ denotes the $r$-radius neighborhood of the vector $x$. For any set \(X\), we define \(\textsf{boundary}(X)\) and \(\textsf{int}(X)\) to be the boundary and interior of set \(X\), respectively. Finally, for any function \(f(\cdot)\), we denote the domain of the function by \(\textsf{dom}(f)\). For any vector \(x\in \mathbb{R}^n\), we define \(\textsf{diag}(x)\in \mathbb{R}^{n\times n}\) to be a diagonal matrix with diagonal entries corresponding to \(x\). 

\section{Model}\label{sec: Model}
We introduce both atomic and non-atomic static games in Sec. \ref{subsec:static}. In Sec. \ref{ssec: DynamicsModel}, we present our proposed adaptive incentive design approach.
\subsection{Static Games}\label{subsec:static}
\subsubsection{Atomic Games}
Consider a game $\finGame$ with a finite set of players $\playerSet$. 
The strategy of each player $i \in \playerSet$ is denoted by $\strategyFin_i \in \strategySetFin_i$, 
where $\strategySetFin_i$ is a {non-empty, closed interval in $\mathbb{R}$}. 
The joint strategy profile of all players is given by $\strategyFin = (\strategyFin_i)_{i \in \playerSet}$, 
and the set of all joint strategy profiles is $\strategySetFin \defas \prod_{i \in \playerSet} \strategySetFin_i$. 
The cost function of each player $i \in \playerSet$ is represented as $\ell_i: \mathbb{R}^{|\playerSet|} \to \mathbb{R}$.

A \textit{system operator} designs incentives by setting a payment $p_i x_i \in \mathbb{R}$ for each player $i$, which is linear in their strategy $x_i$. Here, $\incentiveFin_i \in \mathbb{R}$ represents the marginal payment for every unit increase in the strategy of player $i$. The value of $\incentiveFin_i x_i$ can be either negative or positive, representing a marginal subsidy or a marginal tax, respectively. 
Given the incentive mechanism $p = (p_i)_{i \in I}$, the total cost for player \(i \in \playerSet\) is:
\begin{align}\label{eq: TotCost}
    \costFin_i(\strategyFin,\incentiveFin) = \lossFin_i(\strategyFin) + \incentiveFin_i \strategyFin_i, \quad \forall \ \strategyFin \in \strategySetFin.
\end{align}

A strategy profile $\xEqFin{}(p) \in \strategySetFin$ is a
\textit{Nash equilibrium} in the atomic game $\finGame$ with the incentive mechanism $\incentiveFin$ if  
\begin{align*}
    \costFin_i(\xEqFin{i}(p),\xEqFin{-i}(p),\incentiveFin) \leq \costFin_i(\strategyFin_i,\xEqFin{-i}(p),\incentiveFin), ~\forall \ \strategyFin_i \in \strategySetFin_i,~ \forall i \in \playerSet. 
\end{align*}
A strategy profile $\socOptFin \in \strategySetFin$ is \textit{socially optimal} if it minimizes the social cost function \(\socCostFin: \mathbb{R}^{|\playerSet|} \to \mathbb{R}\) over \(\strategySetFin\).

{
\smallskip 
\begin{assm}\label{assm: SocCostAtomic}
For any \( p \in \mathbb{R}^{|\playerSet|} \), the Nash equilibrium \( x^\ast(p) \) is unique and Lipschitz continuous in \( p \).  
Moreover, the social cost function \( \socCostFin(x) \) is continuously differentiable, has a Lipschitz gradient, and is strictly convex in \( x \). 
\end{assm}

Assumption \ref{assm: SocCostAtomic} is widely adopted in the literature to study incentive design in atomic games (e.g., \cite{li2024socially, liu2021inducing, li2023inducing, shakarami2022dynamic}), either directly or through other conditions that guarantee this\footnote{Uniqueness and Lipschitz continuity of \( x^\ast(p) \) hold if, for every \( i \in \playerSet \) and \( \strategyFin_{-i} = (\strategyFin_j)_{j \in \playerSet \setminus \{i\}} \), the cost function \( \ell_i(\strategyFin_i, \strategyFin_{-i}) \) is strongly convex in \( \strategyFin_i \) and \( \ell_i(\cdot) \) is continuously differentiable with a Lipschitz gradient \cite{dafermos1988sensitivity}.}. 
}
\subsubsection{Non-atomic Games}\label{ssec: PopGame} 
Consider a game \( \popGame \) with a finite set of player populations \( \pop \). Each population \( i \in \pop \) is comprised of a continuum set of (infinitesimal) players with mass \( \massPop_i >0 \). Every (infinitesimal) player in any population can choose an action in a finite set \( \stratPop_i \). The strategy distribution of population \( i \in \pop \) is \( \strategyPop_i = (\strategyPop_i^j)_{j \in \stratPop_i} \), where \( \strategyPop_i^j \) is the mass of players in population \( i \) who choose action \( j \in \stratPop_i \). 
Then, the set of all strategy distributions of population $i$ is $\strategySetPop_i =\left\{\strategyPop_i\in \R^{|\stratPop_i|}| \sum_{j \in \stratPop_i} \strategyPop_i^j = \massPop_i, ~ \strategyPop_i^j \geq 0, \forall j \in \stratPop_i\right\}$. The strategy distribution of all populations is given by $\strategyPop=(\strategyPop_i)_{i \in \pop} \in \strategySetPop = \prod_{i \in \pop} \strategySetPop_i$. We define \(\stratPop=\prod_{i\in\pop}\stratPop_i\). Given a strategy distribution $\strategyPop\in \strategySetPop$, the cost of players in population $i \in \pop$ for choosing action $j \in \stratPop_i$ is $\lossPop_{i}^{j}(\strategyPop)$. We denote $\lossPop_i(\strategyPop) = (\lossPop_i^j(\strategyPop))_{j \in \stratPop_i}$ as the vector of costs for each $i \in \pop$.

A system operator designs incentives by setting a payment \(\incentivePop{i}{j}\) 
for players in population \(i\) who choose action \(j\in\stratPop_i\). 
Consequently, given the incentive mechanism \(\incentivePop{}{} = ({\incentivePop{i}{j}})_{j\in S_i, i\in\pop}\), 
the total cost experienced by any player in population \(i \in \pop\) who chooses action \(j \in \stratPop_i\) is
\begin{align}\label{eq: TotCostPop}
    \costPop{i}{j}(\strategyPop,\incentivePop{}{}) = \lossPop_i^{j}(\strategyPop)+\incentivePop{i}{j}, \quad \forall \  \strategyPop\in \strategySetPop.
\end{align}

A strategy distribution $\xEqPop{}{}(\ptilde) \in \strategySetPop$ 
is a Nash equilibrium in the non-atomic game $\popGame$ with
$\incentivePop{}{}$ if 
\begin{equation}\label{eq: NEPop}
\begin{aligned}
    \xEqPop{i}{j}(\ptilde) &> 0, \quad \Rightarrow \costPop{i}{j}(\xEqPop{}{}(\ptilde),\incentivePop{}{})\leq \costPop{i}{j'}(\xEqPop{}{}(\ptilde),\incentivePop{}{}), \\
    & \hspace{2cm} ~\forall j, j'\in \stratPop_i, ~\forall i \in \pop.
\end{aligned}
\end{equation}
A strategy distribution $\socOptPop \in \strategySetPop$ is socially optimal if $\socOptPop$ minimizes a social cost function $\socCostPop: \R^{|\stratPop|} \to \mathbb{R}$. 

\smallskip

{
\begin{assm}\label{assm: MonotonicCostPop}
For any \(p\in \R^{|\tilde{\playerSet}|}\), the Nash equilibrium \(\tilde{x}^\ast(\tilde{p})\) is unique and Lipschitz continuous in \(\tilde{p}\). Moreover, \(\socCostPop(\tilde{x})\) is continuously differentiable and strictly convex.
\end{assm}

Assumption \ref{assm: MonotonicCostPop} is widely adopted in the literature on incentive design for non-atomic games (e.g., \cite{liu2021inducing, ochoa2022high, mojica2022stackelberg}), either directly or through other conditions that guarantee this\footnote{Uniqueness and Lipschitz continuity of \(x^\ast(p)\) hold if \(\tilde{\ell}(\cdot)\) is Lipschitz continuous and strongly monotone \cite{sandholm2010population}. That is, there exists \(\rho > 0\) such that  \(
    \langle \tilde{\ell}(\tilde{x}) - \tilde{\ell}(\tilde{x}'),\tilde{x}-\tilde{x}' \rangle \geq \rho\|\tilde{x}-\tilde{x}'\|^2\) for every \(\tilde{x}\neq \tilde{x}'\in \tilde{X}.\)
}. 
}

\subsection{Coupled Strategy and Incentive Update}\label{ssec: DynamicsModel}
We consider a coupled dynamical system that jointly updates players' strategies and the incentive mechanism with discrete time-steps \(k\in \mathbb{N}\). At step \(k\), the strategy profile in the atomic game \(\G\) (resp. non-atomic game \(\Gtilde\)) is \(\xk{k}= (\xki{k})_{i \in \playerSet}\) (resp. \(\xtildek{k}= (\xtildeki{k})_{i \in \pop}\)), where \(\xki{k}\) (resp. \(\xtildeki{k}\)) is the strategy of player \(i\) (population \(i\)), and the incentive mechanism is \(\pk{k}= (\pki{k})_{i \in \playerSet}\) (resp. \(\ptildek{k} = (\ptildeki{k}^j)_{j \in S_i, i \in \pop}\)). The strategy updates and the incentive updates are presented below:

\medskip 
\noindent\textbf{Strategy update.} 
\begin{align}
   \xk{k+1}&= (1-\stepx{k})\xk{k}+\stepx{k}\f(\xk{k},\pk{k}),\tag{$x$-update}\label{eq: FinUpdateX}\\
    \xtildek{k+1} &= (1-\stepx{k})\xtildek{k}+\stepx{k}\tilde{f}(\xtildek{k},\ptildek{k}).\tag{$\xtilde$-update}\label{eq: PopUpdateX}
\end{align}
In each step $k+1$, the updated strategy is a linear combination of the strategy in stage $k$ (i.e. $\xk{k}$ in $\G$ and $\xtildek{k}$ in $\Gtilde$), and a new strategy $\f(\xk{k},\pk{k}) \in \X$ in $\G$ (resp. $\ftilde(\xtildek{k},\ptildek{k}) \in \Xtilde$ in $\Gtilde$) that depends on the previous strategy and the incentive mechanism. The relative weight in the linear combination is determined by the step-size $\stepx{k} \in (0, 1)$. 
We require that for any $p$ (resp. $\tilde{p}$), the fixed point associated with update \eqref{eq: FinUpdateX} (resp. \eqref{eq: PopUpdateX}) is a Nash equilibrium, i.e. 
\begin{equation}\label{eq: FixedPointUpdates}
\begin{aligned}
    x^\ast(p)&=\{x: f(x,p)=x\},\quad \forall  p\in \R^{|\mathcal{I}|}, \\
    \tilde{x}^\ast(\tilde{p}) &= \{\tilde{x}: \tilde{f}(\tilde{x},\tilde{p}) = \tilde{x}\}, \quad \forall \ \tilde{p}\in \R^{|\pop|}.
\end{aligned}
\end{equation}  
{We shall impose additional assumptions on \(f(\cdot)\) and \(\tilde{f}(\cdot)\) when studying the convergence of strategy and incentive updates in the next section.
}
Some examples of commonly studied learning dynamics \eqref{eq: FinUpdateX} and \eqref{eq: PopUpdateX} include: 
\begin{enumerate}
    \item \textit{Equilibrium update (\cite{cui2024learning, jalota2022online}):}
        The strategy update incorporates a Nash equilibrium strategy profile with respect to the incentive mechanism in step $k$:
        \begin{align}\label{eq: EqUpdate}
            \f(\strategyFin_k,\incentiveFin_k) = \xEqFin{}(\incentiveFin_k), \quad \text{and}~\tilde{f}(\strategyPop_k,\incentivePop{k}{})=\xEqPop{}{}(\incentivePop{k}{}).
        \end{align}
        \item \textit{Best response update (\cite{fudenberg1998theory,sandholm2010population}):} The strategy update incorporates a best response strategy with respect to the strategy and the incentive mechanism in step $k$:
        \begin{equation}\label{eq: BestResponse}
            \begin{aligned}
                \f_i(\xk{k},\pk{k}) 
     &= \underset{y_i\in\strategySetFin_i}{\arg\min} \ \costFin_i(y_i,x_{-i,k},\pk{k}), \\ 
    \tilde{f}_i(\xtilde_{k},\ptildek{k})&= \underset{\tilde{y}_i\in\strategySetPop_i}{\arg\min}  \  \tilde{y}_i^\top \costPop{i}{}(\tilde{x}_{k},\ptildek{k}),
            \end{aligned}
        \end{equation}
       {{} where the first equation is the best response update in atomic games \cite{fudenberg1998theory}, and the second is the best response update in non-atomic games \cite{sandholm2010population}.}
        \item {\textit{Gradient-based update (\cite{littlestone1994weighted,laraki2013higher,sandholm2010population})}: Gradient-based strategy update commonly studied in literature takes the following form: 
        }
        \begin{equation}\label{eq: GradientBased}
            \begin{aligned}
                \f_i(\xk{k},\pk{k}) &= \underset{y_i\in X_i}{\arg\max} \ z_i(x_k,p_k) y_i - h(y_i), \\ 
    \tilde{f}_i(\xtilde_{k},\ptildek{k}) &=\underset{\tilde{y}_i\in \tilde{X}_i}{\arg\max} \ {{} \tilde{y}_i^\top}\tilde{c}_i(\tilde{x}_k,\tilde{p}_k) - \tilde{h}(\tilde{y}_i),
            \end{aligned}
        \end{equation}
        where \(z_{i}(x_k,p_k)= x_k - \eta \Der_{x_i}c_i(x_k,p_k)\), \(\eta\) is step size, and \(h(\cdot),\tilde{h}(\cdot)\) are regularizers. If \(h(\cdot)\) is a quadratic function, then the update becomes projected gradient descent \cite{mertikopoulos2019learning}. Furthermore, if \(\tilde{h}(\cdot)\) is the entropy function, then the update becomes a perturbed best-response update\cite{sandholm2010population}.
        \end{enumerate}

\medskip 
\noindent\textbf{Incentive update.}
\begin{align}
\pk{k+1} &= (1-\stepp{k})\pk{k}+\stepp{k} \mdFin(\xk{k}), \tag{$p$-update}\label{eq: FinUpdateP}\\
   \ptildek{k+1} &= (1-\stepp{k})\ptildek{k}+\stepp{k} \mdPop(\xtildek{k}), \tag{$\ptilde$-update}\label{eq: PopUpdateP}
\end{align}
where $e(x) = (e_i(x))_{i \in \playerSet}$, $\tilde{e}(\tilde{x}) = (\tilde{e}_i^j(\tilde{x}))_{j \in S_i, i \in \tilde{\playerSet}}$, and 
\begin{subequations}
\begin{align}
    \externality_i(\strategyFin) &= \Der_{\strategyFin_i} \socCostFin(\strategyFin) -  \Der_{\strategyFin_i} \ell_i(\strategyFin), \quad \forall i \in \playerSet, \label{eq: ExterFin} \\
     \externalityPop{i}{j}(\strategyPop) &= \Der_{\strategyPop_i^{j}} \socCostPop(\strategyPop) - \lossPop_{i}^{j}(\strategyPop), \quad \forall j \in \tilde{S}_i, \quad \forall i \in \pop.\label{eq: ExterPop}
\end{align}
\end{subequations}
In \eqref{eq: ExterFin}, $\externality_i(\strategyFin)$ represents the difference between the marginal social cost and the marginal cost of player \(i\) given \(\strategyFin\). Similarly, $\externalityPop{i}{j}(\strategyPop)$ denotes the difference between the marginal social cost and the cost experienced by players in population \(i\) who choose action \(j\). We refer to $\externality_i(\strategyFin)$ and $\tilde{e}_i(\tilde{x})= (\externalityPop{i}{j}(\strategyPop))_{j\in S_i}$ as the externalities of players \(i\) and population \(i\), respectively, since they capture the difference in the impact of their strategies on the social cost and individual cost.

 The updates \eqref{eq: FinUpdateP}-\eqref{eq: PopUpdateP} modify the incentives on the basis of the externality caused by the players. In each step $k+1$, the updated incentive mechanism is a linear combination of the incentive mechanism in step $k$ (i.e. $\pk{k}$ in $\G$ and $\ptildek{k}$ in $\Gtilde$), and the externality (i.e. $\mdFin(\xk{k})$ in $\G$ and $\mdPop(\xtildek{k})$ in $\Gtilde$) given the strategy in step $k$. The relative weight in the linear combination is determined by the step size $\stepp{k} \in (0, 1)$. 

In summary, the joint evolution of strategy and incentive mechanism $(\xk{k}, \pk{k})_{k=1}^{\infty}$ (resp. $(\xtildek{k}, \ptildek{k})_{k=1}^{\infty}$) in the atomic game $\G$ (resp. non-atomic game $\Gtilde$) is governed by the learning dynamics \eqref{eq: FinUpdateX}--\eqref{eq: FinUpdateP} (resp. \eqref{eq: PopUpdateX}--\eqref{eq: PopUpdateP}).  The step-sizes $(\stepx{k})_{k=1}^{\infty}$ and $(\stepp{k})_{k=1}^{\infty}$ determine the speed of strategy updates and incentive updates, respectively. 

\medskip 
{{} \noindent\textbf{Information environment of incentive update.}
The incentive updates in \eqref{eq: FinUpdateP} and \eqref{eq: PopUpdateP} are based on the externalities created by players' strategies. In the absence of additional problem structure, computing externality requires oracle access to  the gradient of players' costs (first-order information) in atomic games (cf. \eqref{eq: ExterFin}) or the players' cost functions (zeroth-order information) in non-atomic games (cf. \eqref{eq: ExterPop}), both evaluated at the current strategy profile. This information requirement is less demanding compared to the gradient-based incentive updates adopted in previous literature \cite{mojica2022stackelberg, chen2023high, poveda2017class, ochoa2022high, liu2021inducing, li2023inducing, alpcan2009nash, alpcan2009control}, where estimating the gradient of the social cost function with respect to the incentive vector often requires the knowledge of the game Jacobian (i.e., second-order information) \footnote{{{} For instance, the gradient based incentive update of atomic games studied in
  \cite{liu2021inducing}
  takes the following form \(
      p_{k+1} = p_k - \beta_k \nabla_p x^\ast(p_k)^\top \nabla_x\Phi(x^\ast(p_k)),\)
  which is a gradient descent update on the function \(\Phi(x^\ast(p))\). The authors estimate \(\nabla_p x^\ast(p_k)^{\top}\) with \(-\nabla_{p}J(x_k;p_k)^\top(\nabla_x J(x_k;p_k))^{-1} \), where \(J(x;p) = (\partial c_i(x,p)/\partial x_i)_{i\in \mathcal{I}}\) is the \textit{game Jacobian}. Therefore, these updates require second order information about the cost function of players. Meanwhile, our approach of externality based pricing only requires first-order information about the cost function of players.}} \cite{liu2021inducing, li2023inducing}, {or knowledge of the gradient of equilibrium strategy with respect to incentive (i.e., \(\nabla_p x^\ast(p)\))  \cite{alpcan2009control, alpcan2009nash}}. Furthermore, our incentive updates do not require knowledge of players' entire cost function, and are agnostic to the specific strategy update dynamics (i.e., \eqref{eq: FinUpdateX} and \eqref{eq: PopUpdateX}) employed by the players. 
  
In many settings, leveraging the structure of the underlying problem enables the social planner to compute externalities with less information.  
For instance, in non-atomic routing games (see Section \ref{sec:routing}), the social planner can compute externality using only the travel time costs of edges (road segments in the network) instead of the cost of each path taken by each population given their origin-destination pair. Additionally, in energy system applications (e.g., \cite{li2024socially}), player's cost function \(\ell_i(x) = g_i(x_i)\) often only depends on their own energy consumption $x_i$, and the social cost function \(\Phi(x) = r(x) + \sum_{i \in \mathcal{I}} g_i(x_i)\) is modeled as the sum of the public cost \(r(x)\) that depends on the joint action \(x\) and the cost of individual players. 
In this case, the externality for any player \(i\) depends only on the gradient of the public cost function \(r(x)\) and not on the private cost of players.
\[
    e_i(x) = \frac{\partial \Phi(x)}{\partial x_i} - \frac{\partial \ell_i(x)}{\partial x_i} = \frac{\partial r(x)}{\partial x_i}.
\]  
}

\section{General results}\label{sec: Results}
In Section \ref{subsec:fpa}, we characterize the set of fixed points of the updates \eqref{eq: FinUpdateX}-\eqref{eq: FinUpdateP} and \eqref{eq: PopUpdateX}-\eqref{eq: PopUpdateP}, and show that any fixed point corresponds to a socially optimal incentive mechanism such that the induced Nash equilibrium strategy profile minimizes the social cost. In Section \ref{subsec: Convergence}, we provide a set of sufficient conditions that {{} guarantee (local and global) convergence} of incentive updates. Under these conditions, our adaptive incentive mechanism eventually induces a socially optimal outcome.  

\subsection{Fixed point analysis}\label{subsec:fpa}
We first characterize the set of fixed points of the updates \eqref{eq: FinUpdateX}-\eqref{eq: FinUpdateP}, and \eqref{eq: PopUpdateX}-\eqref{eq: PopUpdateP} as follows: 
\begin{subequations}
\begin{align}
   &\text{Atomic game $\G$,} ~ \left\{(x, p) | f(x, p) = x, ~ e(x) = p\right\}, \label{subeq:fixedfin}\\
   &\text{Non-atomic game $\Gtilde$,} ~ \left\{(\tilde{x}, \tilde{p}) | \tilde{f}(\tilde{x}, \tilde{p}) = \tilde{x}, ~ \tilde{e}(\tilde{x}) = \tilde{p}\right\}.\label{subeq:fixedpop}
\end{align}
\end{subequations}
Using \eqref{eq: FixedPointUpdates}, from \eqref{subeq:fixedfin} -- \eqref{subeq:fixedpop}, we can write the set of incentive mechanisms at the fixed point, $\Peq$ as follows: 
\begin{equation}
\begin{aligned}\label{eq: PDagger}
   &\text{Atomic game $\G$,} ~ \Peq = \{(\peq_i)_{i\in\playerSet} | \exterFin(\xEqFin{}(\peq)) = \peq \}, \\
    &\text{Non-atomic game $\Gtilde$,} ~ \Ptildeeq = \{(\ptildeeq_i)_{i\in\pop} | \exterPop (\tilde{x}^{*}(\ptildeeq)) = \ptildeeq \}.
\end{aligned}
\end{equation}
That is, at any fixed point, the incentive of each player is set to be equal to the externality evaluated at their equilibrium strategy profile.

Our first result characterizes conditions under which the fixed point set $\Peq$ (resp. $\Ptildeeq$) is non-empty and singleton in $\G$ (resp. $\Gtilde$). Moreover, given any fixed point incentive mechanism $\peq \in \Peq$ and $\ptildeeq \in \Ptildeeq$, the corresponding Nash equilibrium is socially optimal. 
\begin{prop}\label{prop: Alignment}
Let Assumptions \ref{assm: SocCostAtomic} hold and the strategy set \(X\) in an atomic game $G$ be compact. The set \(\Peq\) is a non-empty singleton set. The unique \(\peq \in \Peq\) is socially optimal, i.e. \(\xEqFin{}(\peq) = \socOptFin\). 

Moreover, in a non-atomic game $\popGame$ under Assumptions \ref{assm: MonotonicCostPop}, \(\Ptildeeq\) is a non-empty singleton set. The unique \(\ptildeeq \in \Ptildeeq\) is socially optimal, i.e., \(\xEqPop{}{}(\ptildeeq) = \socOptPop\). 
\end{prop}

\medskip 

{{} 
\noindent\textbf{Advantage of externality-based incentive updates.} 
Proposition \ref{prop: Alignment} demonstrates that the externality-based incentive updates \eqref{eq: FinUpdateP} and \eqref{eq: PopUpdateP} ensure that any fixed point must achieve social optimality. In contrast, the gradient-based incentive update, commonly considered in the literature (e.g., \cite{liu2021inducing, li2023inducing, mojica2022stackelberg, alpcan2009control, alpcan2009nash}), does not guarantee that its fixed point corresponds to a socially optimal incentive mechanism. Typically, these works impose additional assumptions, such as the equilibrium social cost function \(\Phi(x^\ast(p))\) (resp. \(\tilde{\Phi}(\tilde{x}^\ast(\tilde{p}))\)) being strongly convex in the incentive mechanism \(p\) (resp. \(\tilde{p}\)) \cite{liu2021inducing, li2023inducing, mojica2022stackelberg},  {or that the gradient of the equilibrium strategy, \(\nabla_{{p}}{x}^\ast({p})\) (resp. \(\nabla_{\tilde{p}}\tilde{x}^\ast(\tilde{p})\)), with respect to the incentive mechanism \(p\) (resp. \(\tilde{p}\)) is non-singular} \cite{alpcan2009control, alpcan2009nash}, to ensure that the fixed points of the gradient-based update achieve the socially optimal outcome.
In fact, in Appendix \ref{sec: AppendixCounterexample}, we provide an example of a two-link non-atomic routing game, where these assumptions are not satisfied and nearly all fixed points of the gradient-based incentive update fail to achieve the socially optimal outcome. Consequently, the gradient-based incentive update can lead to inefficient outcomes. In contrast, our externality-based incentive update has a unique fixed point that always induces the socially optimal outcome.

}

\medskip 
\noindent\emph{Proof of Proposition \ref{prop: Alignment}.}
First, we show that \(\Peq\) is non-empty, i.e., there exists \(\pEqFin{}\) such that \(\exterFin(\xEqFin{}(\pEqFin{}))=\pEqFin{}\). Define the function \(\theta(p) = \exterFin(\xEqFin{}(p))\). By Assumption \ref{assm: SocCostAtomic}, \(\theta\) is well-defined. Thus, the problem reduces to proving the existence of a solution to \(p= \theta(p)\).

We note that Assumption \ref{assm: SocCostAtomic} ensures that \(\theta(p)\) is a continuous function.
Now, 
define \(K\defas \{\theta(p): p\in \R^{|\playerSet|}\}\subseteq \R^{|\playerSet|}\). We claim that the set \(K\) is compact.
Indeed, this follows from two observations. First, the externality function \(\exterFin(\cdot)\) is continuous. Second,  the range of the function \(\xEqFin{}(\cdot)\) is \(\strategySetFin\), which is a compact set.  These two observations ensure that \(\theta(p)= \exterFin(\xEqFin{}(p))\) is a bounded function. 
Let \(\tilde{K}\defas \textsf{conv}(K)\) be the convex hull of \(K\), which in turn is also a compact set. Let's denote the restriction of function \(\theta\) on the set \(\tilde{K}\) as \(\theta_{|\tilde{K}}:\tilde{K}\ra \tilde{K}\) where \(\theta_{|\tilde{K}}(p)=\theta(p)\) for all \(p\in\tilde{K}\). We note that \(\theta_{|\tilde{K}}\) is a continuous function from a convex compact set to itself and therefore, the Schauder fixed point theorem ensures that there exists \(\pEqFin{}\in\tilde{K}\) such that \(\pEqFin{} =\theta_{|\tilde{K}}(\pEqFin{})= \theta(\pEqFin{})\) \cite{smart1980fixed}. This concludes the proof of the existence of \(\pEqFin{}\).
Analogous argument applies for the non-atomic game \(\popGame\) to show that \(\tilde{P}^\dagger\) is non-empty.  

Next, we show that the incentive \(\pEqFin{}\) aligns the Nash equilibrium with socially optimal strategy (i.e. for any \(\pEqFin{}\in\Peq\), \(\xEqFin{}(\pEqFin{})=\socOptFin\)). 
For any \(\pEqFin{} \in \Peq\) and any \(i \in \playerSet\), it holds that \(\pEqFin{i} = \exterFin_i(\xEqFin{}(\pEqFin{}))\).
 This implies that
\(
    \Der_{\strategyFin_i} \ell_i(\xEqFin{}(\pEqFin{}))+\pEqFin{i} = \Der_{\strategyFin_i} \socCostFin(\xEqFin{}(\pEqFin{}))\) for every \(i\in\playerSet\), and thus
    \begin{align}
       \label{eq: AlignmentSocNash} \JacobianIncentiveFin(\xEqFin{}(\pEqFin{}),\pEqFin{}) = \grad \socCostFin(\xEqFin{}(\pEqFin{})),
    \end{align}
    where \(\JacobianIncentiveFin(x,p)\) is the \textit{game Jacobian} defined as \(\JacobianIncentiveFin_i(x,p) = \Der_{\strategyFin_i} \ell_i(x)+p_i\) for every \(i\in\playerSet\).  
From Assumption \ref{assm: SocCostAtomic} and the first order necessary condition for Nash equilibrium \cite{facchinei2007finite}, {we know that the Nash equilibrium \(\xEqFin{}(\pEqFin{})\) must satisfy  
\begin{align}\label{eq: NASHVIProof}
\langle \JacobianIncentiveFin(\xEqFin{}(\pEqFin{}),\pEqFin{}), \strategyFin- \xEqFin{}(\pEqFin{})  \rangle \geq 0, \quad \forall \ \strategyFin \in \strategySetFin. 
\end{align}}
From \eqref{eq: AlignmentSocNash} and \eqref{eq: NASHVIProof}, we observe that 
\begin{align}\label{eq: SOCFINProof}
\langle \nabla \socCostFin(\xEqFin{}(\pEqFin{})), \strategyFin-\xEqFin{}(\pEqFin{})  \rangle \geq 0, \quad \forall \ \strategyFin \in \strategySetFin.
\end{align}
Further, from the first order conditions of optimality for social cost function we know that \(x^\dagger\) is socially optimal if any only if it satisfies 
{\begin{align}\label{eq: SocCostFiniteVI}
    \langle \nabla \Phi(x^\dagger), x-x^\dagger \rangle  \geq 0, \quad \forall \ x\in X.
\end{align}}Comparing \eqref{eq: SOCFINProof} with \eqref{eq: SocCostFiniteVI}, we note that \(\xEqFin{}(\pEqFin{})\) is the minimizer of social cost function \(\socCostFin\). This implies that \(\xEqFin{}(\pEqFin{}) = \socOptFin\), since \(\socOptFin\) is the unique minimizer of the social cost function \(\socCostFin\) under Assumption \ref{assm: SocCostAtomic}.

Similarly, for non-atomic game \(\popGame\), we show that the incentive \(\pEqPop{}{}\) aligns the Nash equilibrium with social optimality. Fix \(\pEqPop{}{}\in \Ptildeeq\). For every \(j\in \tilde{S}_i\) and \(i\in \pop\), it holds that \(\pEqPop{i}{j} = \exterPop_{i}^{j}(\xEqPop{}{}(\pEqPop{}{}))\). Consequently, 
\begin{align}\label{eq: ExternalityEqnon-atomic}
    \tilde{c}_i^j(\xEqPop{}{}(\pEqPop{}{}), \pEqPop{}{}) = \Der_{\tilde{x}_i^j}\socCostPop(\xEqPop{}{}(\pEqPop{}{})). 
\end{align}
Under Assumption \ref{assm: MonotonicCostPop}, 
{\(\xEqPop{}{}(\pEqPop{}{})\) is a Nash equilibrium only if
\begin{align}\label{eq: VarEqnon-atomic}
    \langle \tilde{c}(\xEqPop{}{}(\pEqPop{}{}), \pEqPop{}{}), \tilde{x} -  \xEqPop{}{}(\pEqPop{}{}) \rangle \geq 0, \quad \forall \ \tilde{x} \in \tilde{X}. 
\end{align}}From \eqref{eq: ExternalityEqnon-atomic} and \eqref{eq: VarEqnon-atomic}, we observe that
\begin{align}\label{eq: VarEqnon-atomicSocCost}
    \langle \nabla \socCostPop(\xEqPop{}{}(\pEqPop{}{})), \tilde{x} -  \xEqPop{}{}(\pEqPop{}{}) \rangle \geq 0, \quad \forall \ \tilde{x} \in \tilde{X}. 
\end{align}
Comparing \eqref{eq: VarEqnon-atomicSocCost} 
with the first order necessary and sufficient conditions of optimality of social cost function, 
we note that \(\xEqPop{}{}(\pEqPop{}{})\) is the minimizer of the social cost function \(\socCostPop\). This implies that \(\xEqPop{}{}(\pEqPop{}{}) = \socOptPop\), since \(\socOptPop\) is the unique minimizer of the social cost function \(\socCostPop\) under Assumption \ref{assm: MonotonicCostPop}.

Finally, we show that the set \(P^\dagger\) is singleton. We prove this via contradiction. Suppose that \(P^\dagger\) contains two element \(p^\dagger_1,p^\dagger_2\), and both align the Nash equilibrium with social optimality. Then, \(x^\dagger = x^\ast(p^\dagger_1) = x^\ast(p^\dagger_2)\). From \eqref{eq: PDagger}, we know that  \(p^\dagger_1 = \exterFin(x^\ast(p^\dagger_1))\) and \(p^\dagger_2 = \exterFin(x^\ast(p^\dagger_2))\). Thus, we must have \( p^\dagger_1 = \exterFin(x^\dagger) = p^\dagger_2 \), which implies that \( P^{\dagger} \) is a singleton.
The proof of uniqueness of \(\tilde{P}^\dagger\) follows analogously. \hfill $\square$

\subsection{Convergence to optimal incentive mechanism}\label{subsec: Convergence}
In this subsection, we provide sufficient conditions for the convergence of strategy and incentive updates \eqref{eq: FinUpdateX}-\eqref{eq: FinUpdateP} and \eqref{eq: PopUpdateX}-\eqref{eq: PopUpdateP}. 
Before presenting the convergence result, we first introduce two assumptions. 
\begin{assm}\label{assm: StepSizeAssumption} The step sizes in \eqref{eq: FinUpdateX}-\eqref{eq: FinUpdateP} and \eqref{eq: PopUpdateX}-\eqref{eq: PopUpdateP} satisfy the following conditions:
\begin{itemize}
\item[(i)] $\sum_{k=1}^{\infty}\stepx{k}=\sum_{k=1}^{\infty}\stepp{k}=+\infty$, $ \sum_{k=1}^{\infty}\stepx{k}^2+\stepp{k}^2 < +\infty$. 
\item[(ii)]$\lim_{k\to\infty}{\stepp{k}}/{\stepx{k}}=0$.
\end{itemize}
\end{assm}
Assumption \ref{assm: StepSizeAssumption}-(i) is a standard assumption on step sizes that allows us to analyze the convergence properties of the
discrete-time learning updates through that of a continuous-time dynamical system \cite{borkar2009stochastic}.  {{}  Assumption \ref{assm: StepSizeAssumption}-(ii) ensures that the incentive update evolves on a slower timescale than the players' strategy updates {{} \cite{borkar2009stochastic, lakshminarayanan2017stability}. Any step sizes of the form \(\gamma_k = k^{-a}\) and \(\beta_k = k^{-b}\) with \(0.5<a<b\leq 1,\) satisfy Assumption \ref{assm: StepSizeAssumption}. 

Assumption \ref{assm: StepSizeAssumption} has been adopted in several previous works on adaptive incentive design (e.g., \cite{li2024socially, chandak2024learning}).  
Under Assumption \ref{assm: StepSizeAssumption}, the strategy update \eqref{eq: FinUpdateX} represents a \emph{fast transient}, whereas the incentive update \eqref{eq: FinUpdateP} is a \emph{slow component}. To analyze such discrete-time updates, we employ techniques from two-timescale approximation theory \cite{borkar2009stochastic, borkar2018concentration, chandak2024learning}, which allows us to analyze the convergence of the strategy and incentive updates separately.  
An intermediate step in this process is to ensure that, for every \( p, \tilde{p} \), the trajectories of the following continuous-time strategy dynamics globally converge (cf. \cite{borkar2009stochastic, borkar2018concentration, chandak2024learning}):  
\begin{align}
   \dot{x}(t)&= \f(x(t),p) - x(t),\tag{$x$-dynamics}\label{eq: FinUpdateXCont}\\
    \dot{\tilde{x}}(t) &= \tilde{f}(\tilde{x}(t),\tilde{p})- \tilde{x}(t).\tag{$\xtilde$-dynamics}\label{eq: PopUpdateXCont}
\end{align}
In this work, we do not focus on analyzing the convergence of \eqref{eq: FinUpdateXCont}-\eqref{eq: PopUpdateXCont}. Instead, we assume any off-the-shelf convergent strategy update that satisfies the following assumption:
\begin{assm}\label{assm: ConvergenceStrategy}
For any incentive mechanism \(p\) (resp. \(\tilde{p}\)), the Nash equilibrium \(x^\ast(p)\) (resp. \(\tilde{x}^\ast(\tilde{p})\)) is the globally asymptotically stable fixed point of the continuous-time dynamical system \eqref{eq: FinUpdateXCont} (resp. \eqref{eq: PopUpdateXCont}). 
\end{assm}
Assumption \ref{assm: ConvergenceStrategy} is satisfied for a variety of strategy updates in various games. 
This includes the best-response and fictitious play strategy update in zero-sum and potential games \cite{hofbauer2006best, benaim2005stochastic, swenson2018best, leslie2006generalised}, and gradient-based strategy update in continuous games \cite{mazumdar2020gradient, mertikopoulos2024unified}.

Our goal here is to characterize conditions under
which the coupled strategy and incentive updates \eqref{eq: FinUpdateX}-\eqref{eq: FinUpdateP} and \eqref{eq: PopUpdateX}-\eqref{eq: PopUpdateP}
converge. Before stating the convergence results, we define two notions of convergence. 
}

{{} 
\begin{defn}\label{def: ConvergenceCriterion}
    We say that the coupled strategy and incentive updates \eqref{eq: FinUpdateX}-\eqref{eq: FinUpdateP} 
    \begin{itemize}
\item[(i)] \textit{globally converges to the fixed point \((x^\dagger, p^\dagger)\)}  if, for any initial condition \(p_0\in \mathbb{R}^{|\playerSet|}\) and \(x_0\in X\), and any selection of step sizes that satisfy Assumption \ref{assm: StepSizeAssumption}, the discrete-time updates \eqref{eq: FinUpdateX}-\eqref{eq: FinUpdateP} asymptotically converge to \((x^\dagger, p^\dagger)\). 
\item[(ii)] \textit{locally converges to the fixed point \((x^\dagger, p^\dagger)\)}
if there exist positive scalars \(\bar{r}, \bar{\alpha}, \bar{\beta}, \bar{\gamma}\) such that, when \(p_0 \in \mathcal{B}_{\bar{r}}(p^\dagger)\) and \(x_0 \in \mathcal{B}_{\bar{r}}(x^\ast(p_0))\), the step sizes satisfy Assumption \ref{assm: StepSizeAssumption} and the following condition:  
\begin{align}\label{eq: StepSizeBound}
\sup_{k \in \mathbb{N}} \frac{\stepP_k}{\stepX_k} \leq \bar{\alpha}, \quad  
\sup_{k \in \mathbb{N}} \stepP_k \leq \bar{\beta}, \quad  
\sup_{k \in \mathbb{N}} \stepX_k \leq \bar{\gamma},
\end{align}
then the discrete-time updates \eqref{eq: FinUpdateX}--\eqref{eq: FinUpdateP} asymptotically converge to \((\socOptFin, \peq)\).
    \end{itemize}
    Local and global convergence are analogously defined for the updates \eqref{eq: PopUpdateX}-\eqref{eq: PopUpdateP} in non-atomic games. 
\end{defn}
}

{{} \begin{prop}\label{thm: ConvergenceFin}
Consider the atomic game \(G\) with discrete-time update \eqref{eq: FinUpdateX}-\eqref{eq: FinUpdateP} that satisfy Assumptions \ref{assm: StepSizeAssumption} and \ref{assm: ConvergenceStrategy}. The following requirements provide sufficient conditions for \eqref{eq: FinUpdateX}-\eqref{eq: FinUpdateP} to locally converge to the fixed point\footnote{This result holds even if the updates \eqref{eq: FinUpdateX}-\eqref{eq: FinUpdateP} and \eqref{eq: PopUpdateX}-\eqref{eq: PopUpdateP} are perturbed with square-integrable martingale difference noise \cite{borkar2018concentration}.} \((x^\dagger,p^\dagger)\) in the sense of Definition \ref{def: ConvergenceCriterion}:  
\begin{enumerate}
    \item[(R1)] \(p^\dagger\) is a locally asymptotically stable equilibrium of the following continuous-time dynamical system:
    \begin{align}
        \dot{\incentiveFinSep}(t) = \exterFin(\xEqFin{}(\incentiveFinSep(t))) - \incentiveFinSep(t). \label{subeq:p_fin11}
    \end{align}
    \item[(R2)] The trajectories of the discrete-time updates satisfy the boundedness condition:
    \[
        \sup_{k \in \mathbb{N}} \big(\|x_k\| + \|p_k\|\big) < +\infty.
    \]
\end{enumerate}Furthermore, the sufficient conditions for \eqref{eq: FinUpdateX}-\eqref{eq: PopUpdateP} to globally converge to the fixed point \((x^\dagger,p^\dagger)\) in the sense of Definition \ref{def: ConvergenceCriterion} are (R1') and (R2), where
\begin{enumerate}
    \item[(R1')] \(p^\dagger\) is a globally asymptotically stable equilibrium of the continuous-time dynamical system:
    \begin{align}
        \dot{\incentiveFinSep}(t) = \exterFin(\xEqFin{}(\incentiveFinSep(t))) - \incentiveFinSep(t). \label{subeq:p_fin12}
    \end{align}
\end{enumerate}
Analogous result holds for the non-atomic game \(\popGame\).
\end{prop}
}

Proposition~\ref{thm: ConvergenceFin} states two sets of generic conditions that can be verified when studying convergence in any specific game. In particular, by leveraging results from nonlinear dynamical systems theory, (R1) (or (R1')) can be verified by showing the existence of a Lyapunov function~\cite{sastry2013nonlinear} or by establishing that the dynamical system is cooperative~\cite{hirsch1985systems}; see Lemma~\ref{lem: SufficientConditionsPropMain} in Appendix~\ref{sec: AppProofs}. Additionally, (R2) holds in any game with a compact strategy set. In games with an unbounded strategy set, (R2) can be verified by analyzing the global convergence of continuous-time (`scaled') strategy and incentive dynamics~\cite[Theorem 10]{lakshminarayanan2017stability}. In Section~\ref{sec: Applications}, we verify that the conditions in Proposition~\ref{thm: ConvergenceFin} are satisfied in atomic aggregative games and non-atomic routing games.
}
\medskip 
\noindent\textbf{Proof of Proposition \ref{thm: ConvergenceFin}.}
    Assumption \ref{assm: StepSizeAssumption}-(ii),  allow us to study the 
    convergence of \eqref{eq: FinUpdateX}-\eqref{eq: FinUpdateP} in two stages \cite{borkar1997stochastic, borkar2018concentration}. First, we study the convergence of fast strategy updates, for every fixed value of incentive. Second, we
    study the convergence of slow incentive updates, assuming that the fast strategy updates have converged to the equilibrium.
    
    Formally, to study the convergence of fast strategy updates, we re-write \eqref{eq: FinUpdateX}-\eqref{eq: FinUpdateP} as follows 
    \begin{equation}
    \begin{aligned}\label{eq: FastUpdate_Borkar}
        x_{k+1} &= x_k + \gamma_k \lr{f(x_k,p_k) - x_k}, \\ 
        p_{k+1} &= p_k + \gamma_k \frac{\beta_k}{\gamma_k} \lr{e(x_k)-p_k}. 
    \end{aligned}
    \end{equation}
    Since \(\sup_{k\in\mathbb{N}}(\|x_k\|+\|p_k\|)<+\infty\) (cf. requirement (R2)) and \(\lim_{k\rightarrow \infty}\beta_k/\gamma_k = 0\) (cf. Assumption \ref{assm: StepSizeAssumption}), the term \(\frac{\beta_k}{\gamma_k} \lr{e(x_k)-p_k}\) in \eqref{eq: FastUpdate_Borkar} goes to zero as \(k\rightarrow\infty\). Consequently, leveraging the standard approximation arguments \cite[Lemma 1, Section 2.2]{borkar2009stochastic}, we conclude that  
    the asymptotic behavior of the updates in \eqref{eq: FastUpdate_Borkar} is same as that of the following dynamical system
    \begin{align*}
        \dot{\textbf{x}}(t) = f(\textbf{x}(t), \textbf{p}(t)) - \textbf{x}(t), \quad \dot{\textbf{p}}(t)  =  0. 
    \end{align*}
    Using Assumption \ref{assm: ConvergenceStrategy}, we conclude that 
    \begin{align}
\label{eq: ConvergenceFast}\lim_{k\rightarrow\infty}(\strategyFin_k,\incentiveFin_k) \rightarrow \{(\xEqFin{}(\incentiveFin),\incentiveFin):\incentiveFin\in\R^{|\playerSet|}\}.
    \end{align}
    
Next, to study the convergence of the slow incentive updates, we re-write \eqref{eq: FinUpdateP} as follows 
\begin{align}\label{eq: SlowIncentiveUpdates}
        p_{k+1} = p_k   +\beta_k \lr{e(x^\ast(p_k))-p_k} + \beta_k \lr{e(x_k)-e(x^\ast(p_k))}.
\end{align}
We will show that \((p_k)_{k\in\mathbb{N}}\) will asymptotically follow the trajectories of the following continuous-time dynamics:
\begin{align}
    \dot{\incentiveFinSep}(t) &=  \exterFin(\xEqFin{}(\incentiveFinSep(t))) - \incentiveFinSep(t). 
\end{align}
Note that \(p^\dagger\) is the fixed point of the trajectories of the dynamical system \eqref{subeq:p_fin11} (cf. Proposition \ref{prop: Alignment}). 
Requirement (R1) in Proposition \ref{thm: ConvergenceFin} ensures convergence of \eqref{subeq:p_fin11}.

Let \(D^\dagger\) denote the domain of attraction of \(p^\dagger\) for the dynamical system \eqref{subeq:p_fin11}. 
From the converse Lyapunov theorem \cite{scheidegger1964stability}, we know that there exists a continuously differentiable function \(\bar{V}:D^\dagger\rightarrow \mathbb{R}_+\) such that \(\bar{V}(p^\dagger) = 0\), \(\bar{V}(p)>0\) for all \(p\in D^\dagger\backslash\{p^\dagger\}\) and \(\bar{V}(p)\rightarrow\infty\) as \(p\rightarrow \textsf{boundary}(D^\dagger)\).  
For any \(r>0\), define \(\bar{V}_r = \{p\in \textsf{dom}(\bar{V}):\bar{V}(p)\leq r\}\) to be a sub-level set of \(\bar{V}\).    
There exists \( 0<\bar r'<\bar r\) such that  \(\bar V_{\bar{r}'}\subsetneq \mathcal{B}_{\bar{r}'}(p^\dagger)\subsetneq \mathcal{B}_{\bar{r}}(p^\dagger) \subsetneq \bar V_{\bar{r}}\). 
Additionally, define \(t_0 = 0, t_k  = \sum_{i=1}^{k}\beta_i\) and \(L_k = t_{n(k)}\) where \(n(0)=0\), and  
\begin{align}\label{eq: n_k}
    n(k) = \min\left\{ m\geq n(k-1) : \sum_{j = n(k-1)+1}^{m}\beta_j\geq T \right\}~ \forall k\in\mathbb{N}. 
\end{align}
Here, \(T\) is a positive integer to be described shortly. Furthermore, 
define \(\bar{\textbf{p}}^{(k)}:\mathbb{R}_+\rightarrow \mathbb{R}^{|\playerSet|}\) to be a solution of \eqref{subeq:p_fin11} on \([L_k,\infty)\) such that \(\bar{\textbf{p}}^{(k)}(L_k)  = p_{L_k}\).

To ensure that \(\bar{\textbf{p}}^{(k)}(L_k)\in \textsf{dom}(\bar{V})\) for \(k>0\), we show that for an appropriate choice of \(T\) in \eqref{eq: n_k}, \(p_{L_k}\in \textsf{int}(D^\dagger)\) for every \(k\in\mathbb{N}\).
 From \cite[Theorem IV.1]{borkar2018concentration}, we know that there exists \(K>0\) such that for all \(k\in \mathbb{N},\)
\begin{align*}
    &\|p_k - \bar{\textbf{p}}^{(0)}(t_k)\|\notag  \\&\leq K\lr{\sup_{k}\beta_k + \sup_{k}\gamma_k + \sup_{k}\frac{\beta_k}{\gamma_k} + \sup_{k}\frac{\beta_k}{\gamma_k}\|x_0 - x^\ast(p_0)\|}\notag \\
    &= K\lr{\bar{\alpha}+\bar{\beta} + \bar{\gamma} + \bar{\alpha}\bar{r}} =: \kappa. 
\end{align*}
Consequently, using the triangle inequality, it holds that  
\begin{align}\label{eq: TriangleInequality}
    \|p_k-p^\dagger\| \leq \kappa + \|\bar{\textbf{p}}^{(0)}(t_k)-p^\dagger\|. 
\end{align}

Since \(\bar{V}\) is a Lyapunov function of \eqref{subeq:p_fin11} and  \(\bar{\textbf{p}}^{(0)}(0) = p_0 \in\mathcal{B}_{\bar{r}}(p^\dagger) \subsetneq \bar V_{\bar{r}}\), there exists \(\bar{k}\in\mathbb{N}\) such that for all \(k\geq \bar{k}\), \(\bar{\textbf{p}}^{(0)}(t_k)\in \bar V_{\bar{r}'}\subsetneq \mathcal{B}_{\bar{r}'}(p^\dagger)\). If  we choose \(\kappa < \bar{r}-\bar{r}'\) then, from \eqref{eq: TriangleInequality}, it holds that for all \(k\geq \bar{k}\), \(p_k\in \mathcal{B}_{\bar{r}}(p^\dagger)\). Therefore, if we choose \(T\geq \bar{k}\) in \eqref{eq: n_k}, it holds that 
\begin{align}\label{eq: P_LKInSet}
    p_{L_k} \in \textsf{dom}(\bar V), \quad \forall \ k \in \mathbb{N}. 
\end{align}

Define \(\hat{p}:\mathbb{R}_+\rightarrow \mathbb{R}\) such that, for every \(k\in \mathbb{N}\), \(\hat{p}(t_k) = p_k\) with linear interpolation on \([t_k, t_{k+1}]\). Using the standard approximation arguments from \cite[Chapter 6]{borkar2009stochastic}, it holds that\footnote{ For any \(T\geq \bar{k}\) and \(\delta>0\), there exists \(k(\delta)\) such that \(\hat{p}(t_{k(\delta)}+\cdot)\) form a ``\((T,\delta)\)'' perturbation (cf. \cite{borkar1997stochastic}) of  \eqref{subeq:p_fin11}.} 
\begin{align}\label{eq: TDeltaPerturbation}
  &\sup_{t\in [L_k,L_{k+1}]}  \|\hat{p}(t) - \bar{\textbf{p}}^{(k)}(t) \| \notag \\&\leq \mathcal{O}\left(\sum_{m\geq L_k}\beta_m^2 + \sup_{m\geq L_k}\|x_m-x^\ast(p_m)\|\right).
\end{align}
Using \eqref{eq: ConvergenceFast} and Assumption \ref{assm: StepSizeAssumption}, we conclude that RHS in the above equation goes to zero as \(k\rightarrow \infty\).
Finally, using \eqref{eq: P_LKInSet}, \eqref{eq: TDeltaPerturbation} and \cite[Lemma 2.1]{borkar1997stochastic}, we conclude that \(p_k\rightarrow p^\dagger\) as \(k\rightarrow\infty\). \hfill $\square$

\section{Applications}\label{sec: Applications}
In this section, we study the applicability of the {{}  general results from Section \ref{sec: Results} to study convergence of our} externality-based incentive updates in two practically relevant classes of games: atomic aggregative games, and non-atomic routing games. 
\subsection{Atomic Aggregative Games}\label{sec:aggregate}
Here, we study \textit{quadratic} networked aggregative games \cite{bramoulle2007public,bramoulle2016oxford,shakarami2022dynamic, acemoglu2013aggregate}.
Consider a game \(G\) comprised of a finite set of players $\playerSet$.
The strategy set of every player is the entire real line \(\R\).  
Given the joint strategy profile \(x = (x_i)_{i \in \playerSet}\), the cost of each player \(i \in \playerSet\) is given by  
\begin{align}\label{eq: CostAggregate}
    \ell_i(x) = \frac{1}{2} q_i x_i^2 + \alpha x_i (A x)_i,
\end{align}
where \(A \in \R^{|\playerSet| \times |\playerSet|}\) is the \textit{network matrix}, with \(A_{ij}\) representing the impact of player \(j\)'s strategy on the cost of player \(i\). The parameter \(\alpha > 0\) characterizes the impact of the aggregate strategy on the individual cost of players.  
Moreover, \(q_i > 0\) determines the influence of each player's own strategy on their cost function.   
Without loss of generality, we consider \(A_{ii}=0\) for all \(i\in\playerSet\). 
For notational brevity, we define \(Q=\textsf{diag}((q_i)_{i\in \playerSet}) \in \R^{|\playerSet|\times |\playerSet|}\). 

A system operator designs incentives through a payment \( p_i x_i \) for player \( i \) when choosing strategy \( x_i \). Thus, the total cost of player \( i \) is given by  \(
c_i(x, p) = \ell_i(x) + p_i x_i.
\)
{{} The system operator's cost is
\begin{align}\label{eq: AggregativeSocialCost}
    \Phi(x) = \sum_{i=1}^{n} \frac{1}{2}(x_i - \zeta_i)^2,
\end{align}
where \(\zeta = (\zeta_i)_{i\in I}\in \mathbb{R}^{|\playerSet|}\) denotes the socially optimal strategy. Similar cost function has been considered for systemic risk analysis in financial networks \cite{acemoglu2015networks}.  In Appendix \ref{sssec: AdditionalResults}, we generalize our results for a broader class of social cost functions.}

\begin{prop}\label{thm: UniquenessAggregative}
Suppose \(M := Q + \alpha A\) is invertible. Then, the Nash equilibrium is given by \(\xEq(p) = -M^{-1}p\). Furthermore, the set \(P^\dagger\) is a singleton set. 
\end{prop}

{{}  The proof follows by noting that the game is strongly convex and equilibrium is computed by first order conditions. 
Proof of Proposition \ref{thm: UniquenessAggregative} is provided in Appendix \ref{sssec: ProofAggregative}.}

Next, we provide sufficient conditions to ensure {{} global convergence} of \eqref{eq: FinUpdateX}-\eqref{eq: FinUpdateP} to the fixed points.
{{} 
\begin{prop}\label{prop: PropGlobalAggregative}
    Consider the updates  \eqref{eq: FinUpdateX}-\eqref{eq: FinUpdateP} associated to the aggregative game $G$. Suppose that Assumptions \ref{assm: StepSizeAssumption} and \ref{assm: ConvergenceStrategy}
    are satisfied. 
Additionally, if 
\begin{itemize}
    \item[(i)] \(M:= Q+\alpha A\) is symmetric positive definite, and 
    \item[(ii)] 
    The function \(f_c(x,p) := \frac{1}{c}(f(cx,cp)-cx)\), satisfy \(f_{c} \rightarrow f_\infty\) as \(c\rightarrow\infty,\) uniformly on the compacts, and for every incentive vector \( p \in \mathbb{R}^{|\playerSet|} \), \( x^\ast(p) \) is the globally asymptotically stable fixed point of  
\begin{align}\label{eq: scaledLimitingSys}
    \dot{x}(t) = f_{\infty}(x(t), p),
\end{align}  
where, for any \( x \in X \) and \( p \in \mathbb{R}^{|\playerSet|} \),  \(
f_{\infty}(x, p) = \lim_{c\rightarrow \infty} f_{c}(x, p).\)  
 
\end{itemize}
Then, the discrete-time updates \eqref{eq: FinUpdateX} and \eqref{eq: FinUpdateP} globally converges to the fixed point \((\socOptFin, \peq)\) in the sense of Definition \ref{def: ConvergenceCriterion}.
\end{prop}

We establish Proposition \ref{prop: PropGlobalAggregative} by verifying requirements (R1') and (R2) of Proposition \ref{thm: ConvergenceFin}.  
To verify (R1'), we use Proposition \ref{prop: PropGlobalAggregative}-(i) to show that  \(
V(p) = (p - p^\dagger)^{\top} M^{-\top} (p - p^\dagger)\)
serves as a Lyapunov function candidate for the dynamical system \eqref{subeq:p_fin12}, guaranteeing global convergence.  
Next, we leverage Proposition \ref{prop: PropGlobalAggregative}-(ii) along with \cite[Theorem 10]{lakshminarayanan2017stability} to show that (R2) of Proposition \ref{thm: ConvergenceFin} holds. Proof of Proposition \ref{prop: PropGlobalAggregative} is in Appendix \ref{sssec: AggregativeDynamicsProof}.

Condition (ii) in Proposition \ref{prop: PropGlobalAggregative} and Assumption \ref{assm: ConvergenceStrategy} both impose global convergence of a suitably defined continuous-time strategy dynamics. In general, one need not imply the other. However, the two conditions become equivalent if the strategy update rule \( f(x, p) \) (cf. \eqref{eq: FinUpdateX}) is linear in both \( x \) and \( p \), which is the case if the strategy updates are best-response-based \eqref{eq: BestResponse} or gradient-based \eqref{eq: GradientBased} in aggregative game.

}

\subsection{Non-atomic Traffic Routing on General Networks}\label{sec:routing}
Consider a routing game \(\tilde{G}\) that models the interactions of strategic travelers over a directed graph \(\graph = (\edge, \node)\), where \(\node\) is the set of nodes and \(\edge\) is the set of edges.  
Let \(\odPair\) be the set of origin-destination (o-d) pairs. Each o-d pair \(i \in \odPair\) is connected by a set of routes,\footnote{A route is a sequence of contiguous edges.} denoted by \(\routes_i\). Let \(\routes = \bigcup_{i\in\odPair} \routes_i\) represent the set of all routes in the network.

An infinitesimal traveler on the network is associated with an o-d pair and chooses a route to commute between the o-d pair. Let the total population of travelers associated with any o-d pair \(i\in\odPair\) be denote by \(\massTraffic_i\). Let \(\flowRoute_{i}^j\) be the amount of travelers taking route \(j\in\routes_i\) to commute between o-d pair \(i\in\odPair\) and \(\flowRoute=(\flowRoute_{i}^j)_{j\in\routes_i,i\in\odPair}\) is a vector which contains, as its entries, the route flow of all population on different routes. 
Naturally, for every \(i\in\odPair\), it holds that
\(\sum_{j\in\routes_i}\flowRoute_{i}^j = \massTraffic_i\). Any route flow \(\tilde{x}\) induces a flow on the edges of the network, denoted by \(\tilde{w}\), such that \(\flowEdge_a =  \sum_{i\in\odPair}\sum_{j\in \routes_i} \flowRoute_{i}^j\mathbbm{1}(a\in j),\) for every \(a\in \edge\). We denote the set of feasible route flows by \(\tilde{X}\) and the set of feasible edge flows by \(\tilde{W} = \{(\tilde{w}_a)_{a\in\edge}: \exists \tilde{x}\in \tilde{X}, \flowEdge_a = \sum_{i\in\odPair}\sum_{j\in\routes_i} \tilde{x}_{i}^j\}\). 
For any o-d pair \(i\in\odPair\) and route flow
\(\flowRoute\in \R^{|\routes|}\), the cost experienced by travelers using route \(j\in \routes_i\) is
\(\lossRoute_{i}^{j}(\tilde{x}) = \sum_{a\in \edge}\lossEdge_a(\flowEdge_a)\mathbbm{1}(a\in j)\), where \(\lossEdge_a(\cdot)\) is the \textit{edge latency function} that depends on the edge flows. For every edge \(a\in \edge\), we assume that the edge latency function \(\lossEdge_a(\cdot)\) is convex and strictly increasing. This property of edge latency function captures the congestion effect on the transportation network \cite{beckmann1956studies,roughgarden2010algorithmic}.   
A system operator designs incentives by setting tolls on the edges of the network  in the form of edge tolls\footnote{{}  If we directly use the setup of non-atomic games presented in Section \ref{ssec: PopGame}, we would require the system operator to use \textit{route-based} tolls rather than \textit{edge-based} tolls. Our approach of using edge-based tolls is rooted in practical consideration with implementation of tolls.}, denoted by \(\edgeTolls = (\edgeTolls_a)_{a\in \edge}\). Every edge toll vector induces a unique route toll vector \(\routeTolls\). That is, for any o-d pair \(i\in \odPair\), the toll on route \(j\in\routes_i\) is 
\begin{align}\label{eq: RelEdgeRouteToll}
\routeTolls_{i}^j = \sum_{a\in \edge:a\in j} \edgeTolls_a.
\end{align}
Consequently, the total cost experienced by travelers on o-d pair \(i\in\odPair\) who choose route \(j\in \routes_i\) is \(\loss_{i}^j(\flowRoute,\routeTolls) = \lossRoute_{i}^j(\flowRoute) + \routeTolls_{i}^j\). 
Let \(\xWardrop(\routeTolls)\) denote a Nash equilibrium ({{}  also known as Wardrop equilibrium in non-atomic routing games literature}) corresponding to
route tolls \(\routeTolls\). 
Owing to \eqref{eq: RelEdgeRouteToll}, with slight abuse of notation, we shall frequently use 
\(\xWardrop(\routeTolls)\) and \(\xWardrop(\edgeTolls)\) interchangeably.
Typically, the equilibrium route flows can be non-unique but the corresponding edge flows \(\tilde{w}^\ast(\tilde{p})\) are unique. Furthermore, the function \(\tilde{p}\mapsto\tilde{w}^\ast(\tilde{p})\) is a continuous function \cite{yang2005mathematical}. 

The system operator's objective is to design tolls that ensure that the resulting equilibrium minimizes the overall travel time incurred by travelers on the network, characterized as the minimizer of
\begin{align}\label{eq: SocCostTraffic}
\socCostRouting(\tilde{x}) 
    = \sum_{i\in\odPair}\sum_{j\in\routes_i}\flowRoute_{i}^j\lossRoute_{i}^j(\flowRoute).
\end{align}
Note that the optimal route flow can be non-unique but the optimal edge flow, denoted by \(w^\dagger\), is unique \cite{yang2005mathematical}. 

Using the description of travelers' costs, the externality caused by travelers from o-d pair \(i \in \odPair\) using route \(j \in \routes_i\), based on \eqref{eq: ExterPop}, is given by   
\begin{align}\label{eq: ExternalityRoutingRoute}
    \exterPop_{i}^j(\flowRoute) = \sum_{i'\in\odPair}\sum_{j'\in\routes_i}\flowRoute_{i'}^{j'}\frac{\partial \lossRoute_{i'}^{j'}(\flowRoute)}{\partial \flowRoute_{i}^{j}} \stackrel{(a)}{=}\sum_{a\in \edge: a\in j} \nabla \lossEdge_a(\flowEdge_a) \flowEdge_a,
\end{align}
where \((a)\) is due to Lemma \ref{lem: ExternalityExpression} in Appendix \ref{sec: AppProofs}.

{{} 
From \eqref{eq: ExternalityRoutingRoute}, we note that the externality on any route \(j\) is the sum externality on every edge on that route. Therefore, we study the following incentive update, which updates the edge-tolls as follows: 
\begin{align}\label{eq: P_edge_Dyn}
    \tilde{p}_{a,k+1} = (1-\beta_k)\tilde{p}_{a,k} + \beta_k\tilde{e}_a(\tilde{x}_k), \quad \forall \ a\in \edge, 
\end{align}
where \(\tilde{e}_a(\tilde{x}_k)=  \nabla \lossEdge_a(\flowEdge_{a,k}) \flowEdge_{a,k}\), and \(w_{a,k} = \sum_{i\in \tilde{\playerSet}}\sum_{j\in \tilde{R}_i}\tilde{x}_{i,k}^j.\)
Define 
\begin{align*}
\tilde{\textbf{P}}^\dagger = \{(\optToll_a)_{a\in\edge}: \optToll_a = \pertEdge_a(\optToll)\nabla \lossEdge_a(\pertEdge_a(\optToll)), \forall \ a\in \edge\big\},
\end{align*}
to be the fixed point of the joint update \eqref{eq: PopUpdateX}-\eqref{eq: P_edge_Dyn}.}

\begin{prop}\label{thm: FixedPointsRouting}
The set \(\tilde{\textbf{P}}^\dagger\) is non-empty singleton set. The unique \(\peq \in \tilde{\textbf{P}}^\dagger\) is socially optimal, i.e. \(\tilde{w}(\peq) = w^\dagger\).
\end{prop}
{{} 
Proof of Proposition \ref{thm: FixedPointsRouting} follows in two steps. First, we show that any \(\pEqFin{}\in \tilde{\textbf{P}}^\dagger\) aligns the Nash equilibrium with social optimality, i.e. \(\tilde{w}(\peq) = w^\dagger\). Next, using contradiction argument similar to the proof of Proposition \ref{prop: Alignment}, we show that \(\tilde{\textbf{P}}^\dagger\) is singleton. Detailed proof of Proposition \ref{thm: FixedPointsRouting} is provided in Appendix \ref{sssec: RoutingFixedPoint}.}

Next, {{} we provide sufficient conditions for local convergence of} the updates \eqref{eq: PopUpdateX}-\eqref{eq: PopUpdateP}. 
\begin{prop}\label{thm: ConvergenceRouting} {{} Consider the updates  \eqref{eq: PopUpdateX}-\eqref{eq: PopUpdateP} associated with the routing game $\tilde{G}$.
Suppose that Assumptions  \ref{assm: StepSizeAssumption} and \ref{assm: ConvergenceStrategy} are satisfied,} and there exists an equilibrium route flow \(\tilde{x}^\ast(\tilde{p}^\dagger)\) such that for every \(i\in \tilde{I}, j, j'\in \tilde{R}_i,\)
\begin{align}\label{eq: degenerate}
\tilde{c}_i^j(\tilde{x}^\ast(\tilde{p}^\dagger)) \leq  \tilde{c}_i^{j'}(\tilde{x}^\ast(\tilde{p}^\dagger))   \implies  \tilde{x}_{i}^{j,\ast}(\tilde{p}^\dagger) > 0. 
\end{align}
The discrete-time updates \eqref{eq: PopUpdateX}-\eqref{eq: PopUpdateP} locally converges to fixed point \((\tilde{x}^\dagger, \tilde{p}^\dagger)\) in the sense of Definition \ref{def: ConvergenceCriterion}. 

\end{prop}

{{} 
\begin{remark}
There is a subtle distinction between the definition of Nash equilibrium (cf. \eqref{eq: NEPop}) and \eqref{eq: degenerate}.  
The former states that at equilibrium, any route with a positive flow must have the minimum cost.  
In contrast, \eqref{eq: degenerate} further requires that all minimum-cost routes have strictly positive equilibrium flow.  
This regularity condition, commonly used in transportation literature (\cite[Chapter 4]{yang2005mathematical}),  
ensures the differentiability of link flows \(\tilde{w}^\ast(p)\) in the neighborhood of \(\tilde{p}^\dagger\).
\end{remark}

We show Proposition \ref{thm: ConvergenceRouting} by verifying the requirements (R1)-(R2) in Proposition  \ref{thm: ConvergenceFin}. 
(R2) holds due to the fact that \(\tilde{X}\) is compact. Thus, it only remains to verify (R1). 
Towards this goal, we define \(\Delta\in\R^{|\edge|\times|\edge|}\) to be a diagonal matrix such that, for every \(a\in \edge,\)
\begin{align}\label{eq: DeltaMat}
\Delta_{a,a} = (\nabla \lossEdge_a(\pertEdge_a(\optToll)) + \pertEdge_a(\optToll) \nabla^2\lossEdge_a(\pertEdge_a(\optToll)))^{-1}.
\end{align}
Using condition \eqref{eq: degenerate}, we show that \(V(\edgeTolls) = (\edgeTolls-\optToll)^\top \Delta(\edgeTolls-\optToll)\),
acts as a Lyapunov function candidate for the following dynamical system
\begin{align}\label{eq: DynTraffic}
\dot{\tilde{\mathbf{p}}}_a(t) = \tilde{w}^\ast_a(\tilde{\mathbf{p}})\nabla l_a(\tilde{w}^\ast_a(\tilde{\mathbf{p}}))- \tilde{\mathbf{p}}_a,\quad \forall \ a\in \tilde{\mathcal{E}}. 
\end{align}
Detailed proof is provided in Appendix \ref{sssec: ConvergenceRouting}. 
}

\section{Concluding Remarks}\label{sec: Conclusion}
We propose an adaptive incentive mechanism that updates based on agents' externalities, operates independently of their learning rules, and evolves on a slower timescale, forming a two-timescale coupled strategy and incentive dynamics. We show that its fixed point corresponds to an optimal incentive ensuring he Nash equilibrium of the corresponding game achieves social optimality. Additionally, we provide sufficient conditions for convergence of the coupled dynamics and validate our approach in atomic quadratic aggregative games and non-atomic routing games.

There are several interesting directions for future research. One important direction is to verify the convergence conditions in a broader class of games than studied in this paper.  
{Furthermore, it would be valuable to study the design and analysis of externality-based incentive update in scenarios where the social cost depends not only on the agents' strategies but also on the incentive mechanism itself. Additionally, developing adaptive processes for online estimation of game-relevant parameters—necessary for computing externality—by only using agents' strategies would be an interesting research direction.}

\section*{Acknowledgements}
This work is supported by NSF Collaborative Research: Transferable, Hierarchical, Expressive, Optimal, Robust, Interpretable NETworks (THEORINET) under award No. DMS-2031899.

\newpage 
\appendix  
\section{Counter-example.}\label{sec: AppendixCounterexample}
In this section, we present a non-atomic game, where the standard gradient-based incentive design approach would fail. {Specifically, we will show that the gradient of equilibrium strategy with respect to incentive is singular and the equilibrium social cost function is non-convex in the incentive. Furthermore, we show that the fixed points of the gradient-based incentive update is non-unique, and almost all fixed points fail to induce a socially efficient outcome. In contrast, the fixed point of our externality-based incentive update is unique and results in a socially optimal outcome.} 

Consider a non-atomic routing game, comprising of two nodes and two edges connecting them. This network is used by one unit of travelers traveling from the source node \(\mc{S}\) to the destination node \(\mc{D}\). The latency function of two edges are denoted in Figure \ref{fig: SchematicFigure}. 
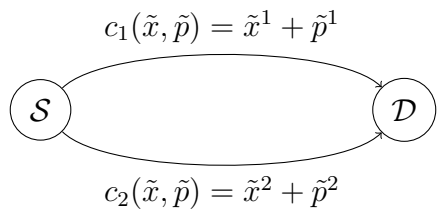
\begin{figure}[h]
\centering 
\begin{tikzpicture}[main/.style = {draw, circle}] 
\node[main] (1) {$\mc{S}$}; 
\node[main] (2) [right= 4cm of 1] {$\mc{D}$}; 
\draw[->] (1) to [out=45,in=135,looseness=0.3] node[midway,above] {$c_1(\tilde{x}, \tilde{p})= \tilde{x}^1+ \tilde{p}^{1}$} (2);
\draw[->] (1) to [out=-45,in=-135,looseness=0.3] node[midway,below] {$c_2(\tilde{x},\tilde{p}) = \tilde{x}^{2} + \tilde{p}^{2}$} (2);
\end{tikzpicture} 
\caption{Two-link routing game.}
\label{fig: SchematicFigure}
\end{figure} 
In this game, the strategy set is $\strategySetPop = \{\strategyPop \in \mathbb{R}^2 : \strategyPop_1+ \strategyPop_2 = 1\}$. The equilibrium congestion levels on the two edges is obtained by computing the minimizer of the following function \cite{yang2005mathematical}: 
\begin{align*}
\potential(\strategyPop,\incentivePop{}{}) = \frac{1}{2}\strategyPop_1^2 + \frac{1}{2} \strategyPop_2^2 + \incentivePop{1}{} \strategyPop_1 + \incentivePop{2}{} \strategyPop_2. 
\end{align*}
Thus, for any toll vector \(\incentivePop{}{}\), the Nash equilibrium \(\xEqPop{}{}(\incentivePop{}{}) = \arg\min_{\strategyPop\in\strategySetPop}\potential(\strategyPop,\incentivePop{}{})\) satisfies  
\(
    \xEqPop{1}{}(\incentivePop{}{}) = \mathcal{P}_{[0,1]}\left( \frac{\incentivePop{2}{} - \incentivePop{1}{} + 1}{2} \right),
    \quad
    \xEqPop{2}{}(\incentivePop{}{}) = \mathcal{P}_{[0,1]}\left( \frac{\incentivePop{1}{} - \incentivePop{2}{} + 1}{2} \right),
\)
where for any scalar \(x \in \mathbb{R}\), \(\mathcal{P}_{[0,1]}(x)\) denotes its projection onto the line segment \([0,1]\).
{The gradient of equilibrium strategy with respect to incentive is a singular matrix for all incentives \(\tilde{p}= (\tilde{p}_1, \tilde{p}_2) \in \mathbb{R}^2\) such that \(|\tilde{p}_1-\tilde{p}_2|>1\). }

The equilibrium social cost function is: 
\begin{align*}
    \socCostPop(\xEqPop{}{}(\incentivePop{}{})) 
    &= \xEqPop{1}{}(\incentivePop{}{}) \lossPop_1( \xEqPop{1}{}(\incentivePop{}{})) +  \xEqPop{2}{}(\incentivePop{}{}) \lossPop_2( \xEqPop{2}{}(\incentivePop{}{}))\\ 
    &=  \begin{cases} \frac{(\incentivePop{1}{}-\incentivePop{2}{})^2 + 1}{2}, &\text{ if } |\incentivePop{1}{} - \incentivePop{2}{}| \leq 1, \\ 1, &\text{otherwise}. \end{cases} 
\end{align*}
{{}  Note that the equilibrium social cost function \(\tilde{\Phi}(\tilde{x}^\ast(\tilde{p}))\) is non-convex in \(\tilde{p}\), which contradicts the assumption commonly adopted in gradient-based incentive learning literature \cite{liu2021inducing, li2023inducing, mojica2022stackelberg}.  Furthermore, the gradient-based update\footnote{Since this function is non-differentiable, it is common to use Clarke's subdifferential to study gradient-based updates \cite{clarke1990optimization}.} for this function takes the following form:
\begin{align*}
    \tilde{p}_{k+1} = \tilde{p}_{k} - \beta_k \partial \Phi(\tilde{x}^\ast(\tilde{p}_k)),
\end{align*}
where   
\begin{align*}
\partial \Phi(\tilde{x}^\ast(\tilde{p})) \in \begin{cases}
        \left\{\begin{bmatrix}
            \tilde{p}_1-\tilde{p}_2 \\ 
            \tilde{p}_2-\tilde{p}_1
        \end{bmatrix}\right\}, & \text{if} \ |\tilde{p}_1-\tilde{p}_2| < 1, \\
       \textsf{conv}\left\{\begin{bmatrix}
            \pm 1 \\ 
            \mp 1
        \end{bmatrix}, \begin{bmatrix}
           0  \\ 
           0
        \end{bmatrix}\right\},  & \text{if} \ \tilde{p}_1-\tilde{p}_2 = \pm 1,
        \\
       \left\{ \begin{bmatrix}
            0 \\ 
            0
        \end{bmatrix} \right\}, & \text{otherwise.}
    \end{cases} 
\end{align*}
Consequently, the set of fixed points for the gradient-based incentive update (i.e., where the gradient is zero) is given by:
\begin{align}\label{eq: Grad_zero}
  \{(\tilde{p}_1,\tilde{p}_2)\in \mathbb{R}^2: |\tilde{p}_1 - \tilde{p}_2| \in \{0\}\cup \{[1,\infty)\}\}.
\end{align}
On the other hand, the set of socially optimal tolls that minimize \(\tilde{\Phi}(\tilde{x}^\ast(\tilde{p}))\) is given by \(\{(\tilde{p}_1,\tilde{p}_2)\in \mathbb{R}^2: \tilde{p}_1 = \tilde{p}_2\}\), which has measure zero within the set of fixed points of the gradient-based update \eqref{eq: Grad_zero}.

In contrast, the fixed point of our externality-based incentive mechanism (cf. \eqref{eq: PDagger}) is unique \(\tilde{p}^\dagger_1=\tilde{p}^\dagger_2= 1/2\) and minimizes the social cost. 
}
{{} 
\section{Proofs and Additional Results on Aggregative Game in Section \ref{sec:aggregate}}
In this section, we present the proofs of Propositions \ref{thm: UniquenessAggregative} and \ref{prop: PropGlobalAggregative}. Additionally, we introduce a generalization of the results in Section \ref{sec:aggregate}.
\subsection{Proof of Proposition \ref{thm: UniquenessAggregative}}\label{sssec: ProofAggregative}
First, we show that \(\xEq(p) = -M^{-1}p\) for any \(p\in\R^{|\playerSet|}\). Note that the cost function \(c_i(x_i,x_{-i},p)\) is strongly convex in \(x_i\) and that the strategy space \(X_i\) is unconstrained, ensuring that the game is strongly convex game. Therefore, \(x^\ast(p)\) is Nash equilibrium if and only if
\(
\nabla_{x_i}c_i(x^\ast(p),p) = 0,\) for every 
\(i\in\playerSet.
\)  
Consequently, using \eqref{eq: CostAggregate}, we obtain  
\begin{align}\label{eq: FOCAggregative}
        q_i x_i^\ast(p) + \alpha (Ax^\ast(p))_i + p_i = 0, \quad \forall i \in \playerSet.
\end{align}  
Stacking \eqref{eq: FOCAggregative} in vector form yields  \(
M x^{\ast}(p) = -p\).

Next, we show that \(P^\dagger\) is a singleton set. Note that 
\begin{align}\label{eq: PdaggerReformulate}
    P^\dagger = \{p^{\dagger}\in \R^{|\playerSet|}: x^\ast_i(p^\dagger) = \zeta_i,~  \ \forall \ i\in \playerSet\}.
\end{align}
The proof concludes by noting that \(x^\ast(p) = -M^{-1}p\). 

\subsection{Proof of Proposition \ref{prop: PropGlobalAggregative}}\label{sssec: AggregativeDynamicsProof}
Here, we verify the requirements (R1') and (R2) of Proposition \ref{thm: ConvergenceFin}. We start with verifying (R1'). We define a Lyapunov function candidate $V(p) = (p-p^\dagger)^{\top} M^{-\top} (p-p^\dagger)$ for the dynamical system \eqref{subeq:p_fin12}.
  Note that \(V(\pEqFin{}) = 0\) and \(V(p)>0\) for all \(p\neq \pEqFin{}\). Next, we show that \(\nabla V(p)^\top(e(x^\ast(p))-p) < 0,\) for every \(p\neq p^\dagger\).  Indeed, 
    \begin{align*}
        &\nabla V(p)^\top(e(x^\ast(p))-p) = 2(p-p^\dagger)^\top M^{-\top}(x^\ast(p)-\zeta) \\
        &= -2(x^\ast(p)-x^\ast(p^\dagger))(x^\ast(p)-x^\ast(p^\dagger)) < 0, \quad \forall \ p \neq p^\dagger.
    \end{align*}
    This shows that \(p^\dagger\) is globally asymptotically stable for \eqref{subeq:p_fin12}. 

Next, for requirement (R2) of Proposition \ref{thm: ConvergenceFin}, we verify sufficient conditions for the boundedness of iterates in two-timescale approximation theory \cite{lakshminarayanan2017stability}. In particular, using \cite[Theorem 10]{lakshminarayanan2017stability}, it is sufficient to show that the following two conditions are satisfied:  
\begin{itemize}
    \item[(a)] The function \(
    f_c(x,p) := \frac{1}{c}(f(cx,cp)-cx)\)
    satisfies \( f_{c} \rightarrow f_\infty \) as \( c \rightarrow \infty \), uniformly on compact sets, for some \( f_{\infty} \).  Also, for every incentive vector \( p \in \mathbb{R}^{|\playerSet|} \), \( x^\ast(p) \) is the globally asymptotically stable fixed point of the following continuous-time dynamical system:  
    \begin{align}\label{eq: scaledFast}
        \dot{x}(t) = f_{\infty}(x(t),p).
    \end{align}
    Furthermore, \( x^\ast(0) = 0 \), and the system \( \dot{x}(t) = f_{\infty}(x(t),0) \) has the origin as its globally asymptotically stable fixed point.  
    \item[(b)] The function  \(
    h_c(p) := \frac{1}{c}(e(cx^\ast(p))-cp)
    \)
    satisfies \( h_c \rightarrow h_{\infty} \) as \( c \rightarrow \infty \), uniformly on compact sets, for some \( h_{\infty} \).  
    Also, the origin is a globally asymptotically stable fixed point of the dynamical system:  
    \begin{align}\label{eq: scaledSlow}
        \dot{p}(t) = h_{\infty}(p(t)).
    \end{align}
\end{itemize}  
Condition (a) is satisfied due to Proposition \ref{prop: PropGlobalAggregative}-(ii) and the fact that \( x^\ast(p) = -M^{-1}p \) in the atomic aggregative game.  
Condition (b) holds since \( h_{\infty}(p) = -(Q+\alpha A)p \). Moreover, since \( Q+\alpha A \) is symmetric positive definite, the origin is a globally asymptotically stable fixed point of \eqref{eq: scaledSlow}.  
}

{{} \subsection{Additional Results}\label{sssec: AdditionalResults}
Here, we consider a more general social cost function than \eqref{eq: AggregativeSocialCost}. Specifically, we consider  
\begin{align}\label{eq: CostFuncGen}
    \Phi(x) = \sum_{i\in \playerSet} h_i(x_i),
\end{align}
where, for every \( i \in \playerSet \), the function \( h_i: \mathbb{R} \to \mathbb{R} \) satisfies the following assumption:
\begin{assm}\label{assm: CostStructureAggregative}
For every \(i\in \playerSet\), \(h_i(\cdot)\) is a strictly convex function with a Lipschitz continuous gradient.  
Furthermore, we assume the existence of \(y^\dagger\in \mathbb{R}^{|\playerSet|}\) such that \(\nabla h_i(y_i^\dagger) = 0\) for every \(i\in \playerSet\).
\end{assm}

\begin{prop}\label{thm: UniquenessAggregativeGen}
Suppose that Assumption \ref{assm: CostStructureAggregative} holds and \(M := Q + \alpha A\) is invertible. Then, the Nash equilibrium \(\xEq(p) = M^{-1}p\) for any \(p\in \mathbb{R}^{|\playerSet|}\). Furthermore, the set \(P^\dagger\) is singleton. 
\end{prop}
\begin{proof}
The proof that \( x^{\ast}(p) = -M^{-1}p \) follows exactly as in Proposition \ref{prop: PropGlobalAggregative}. 
Next, we show that \( P^\dagger \) is a singleton. Using \eqref{eq: ExterFin}, \eqref{eq: CostAggregate}, and \eqref{eq: CostFuncGen}, the externality is given by  
\begin{align}\label{eq: ExternalityAggregative}
    e_i(x) = \nabla h_i(x_i) - q_i x_i - \alpha \sum_{j\in \playerSet} A_{ij}x_j, \quad \forall i \in \playerSet.  
\end{align}
Combining \eqref{eq: FOCAggregative} and \eqref{eq: ExternalityAggregative}, we obtain  
\begin{align}\label{eq: Externality_AggGame}
    e_i(x^\ast(p)) = \nabla h_i(x_i^\ast(p)) + p_i, \quad \forall i \in \playerSet.    
\end{align}
Consequently, using \eqref{eq: PDagger}, we have  
\begin{align}\label{eq: PdaggerReformulate}
    P^\dagger = \{p^{\dagger} \in \mathbb{R}^{|\playerSet|} : \nabla h_i(x^\ast_i(p^\dagger)) = 0,~  \forall i \in \playerSet\}.
\end{align}
Since \( h_i \) is strictly convex, it follows from Assumption \ref{assm: CostStructureAggregative} that there exists a unique \( y^\dagger \) such that \( \nabla h_i(y_i^\dagger) = 0 \) for every \( i \in \playerSet \).  
Therefore, for every \( p^\dagger \in P^\dagger \), it must hold that \( x^\ast(p^\dagger) = y^\dagger \).  
Since \( x^\ast(p) = -M^{-1}p \), it follows that \( p^\dagger = -My^\dagger \), establishing the uniqueness of \( P^\dagger \). 
\end{proof}

{{} Next, we provide sufficient conditions to ensure the convergence of \eqref{eq: FinUpdateX}-\eqref{eq: FinUpdateP} to the fixed points. In particular, we present two sets of conditions: the first set establishes global convergence guarantees, while the second set ensures local convergence guarantees.
}
{{} 
\begin{prop}\label{prop: PropGlobalAggregativeGen}
Consider the updates \eqref{eq: FinUpdateX}-\eqref{eq: FinUpdateP} associated with the aggregative game \( G \). Suppose that Assumptions \ref{assm: StepSizeAssumption}, \ref{assm: ConvergenceStrategy}, and \ref{assm: CostStructureAggregative} hold, and that  
\(\sup_{k\in\mathbb{N}}(\|x_k\|+\|p_k\|)<+\infty\). Additionally,  
\begin{itemize}
    \item[(i)] If \( M := Q+\alpha A \) is symmetric positive definite, then the discrete-time updates \eqref{eq: FinUpdateX} and \eqref{eq: FinUpdateP} globally converge to the fixed point \((\socOptFin, \peq)\) in the sense of Definition \ref{def: ConvergenceCriterion}.  
    \item[(ii)] If \( M := Q+\alpha A \) is invertible with non-negative entries, \( M^{-1} \) has strictly negative off-diagonal entries, and there exists a vector \( y^\dagger \in \mathbb{R}^{|\playerSet|}_- \) such that \( \nabla h_i(y^\dagger_i) = 0 \) for every \( i \in \playerSet \)
    \footnote{A similar statement can be obtained for the case when \( y^\dagger \in \mathbb{R}^{|\playerSet|}_+ \), but we omit it for brevity.}  
    then the discrete-time updates \eqref{eq: FinUpdateX} and \eqref{eq: FinUpdateP} locally converge to the fixed point \((\socOptFin, \peq)\) in the sense of Definition \ref{def: ConvergenceCriterion}.     
\end{itemize}

\end{prop}

\medskip
Propositions \ref{prop: PropGlobalAggregativeGen}-(i) and \ref{prop: PropGlobalAggregative} are related but differ in two key aspects. First, the social cost function in Proposition \ref{prop: PropGlobalAggregative} is a special case of the more general function in \eqref{eq: CostFuncGen}. Second, Proposition \ref{prop: PropGlobalAggregativeGen}-(i) directly assumes boundedness of iterates, \(\sup_{k\in\mathbb{N}}(\|x_k\|+\|p_k\|)<+\infty\), whereas Proposition \ref{prop: PropGlobalAggregative} instead relies on the global convergence of the limiting dynamical system associated with strategy updates (cf. \eqref{eq: scaledLimitingSys}).  
The simpler social cost function \eqref{eq: AggregativeSocialCost} in Proposition \ref{prop: PropGlobalAggregative} allows us to use stability results from two-timescale approximation theory \cite[Theorem 10]{lakshminarayanan2017stability} to establish boundedness. Extending this approach to Proposition \ref{prop: PropGlobalAggregativeGen} would require imposing global convergence of suitably defined limiting dynamical systems (cf. \eqref{eq: scaledFast}-\eqref{eq: scaledSlow}). To maintain clarity, we impose \(\sup_{k\in\mathbb{N}}(\|x_k\|+\|p_k\|)<+\infty\) directly in Proposition \ref{prop: PropGlobalAggregativeGen}.

The conditions imposed on the matrix \(M\) in Proposition \ref{prop: PropGlobalAggregativeGen}-(i) and (ii) are not directly comparable; neither necessarily implies the other \footnote{
 For example, consider aggregative games with parameters \((Q_1, A_1)\) and \((Q_2, A_2)\) such that \(M_1 = Q_1 + \alpha A_1\), \(M_2 = Q_2 + \alpha A_2\) and 
\[
M_1 = \begin{bmatrix}
    1 & 0.1 \\ 
    1 & 1
\end{bmatrix}, 
\quad 
M_2 = \begin{bmatrix}
    1 & -0.1 \\ 
    -0.1 & 1
\end{bmatrix}.
\]
The matrix \(M_1\) satisfies the conditions in Proposition \ref{prop: PropGlobalAggregativeGen}-(ii) but does not satisfy the conditions in Proposition \ref{prop: PropGlobalAggregativeGen}-(i). On the other hand, the matrix \(M_2\) satisfies the conditions in Proposition \ref{prop: PropGlobalAggregativeGen}-(i) but does not satisfy the conditions in Proposition \ref{prop: PropGlobalAggregativeGen}-(ii).}.
}

\medskip

}

\textit{Proof of Proposition \ref{prop: PropGlobalAggregativeGen}:}
 We prove Proposition \ref{prop: PropGlobalAggregativeGen}(i)-(ii) in order. 
\begin{itemize}
\item[(a)] {{} Proposition \ref{prop: PropGlobalAggregativeGen}-(i) follows by verifying the requirements (R1')-(R2) of Proposition \ref{thm: ConvergenceFin}. We only need to verify (R1') as (R2) is satisfied due to the assumption that \(\sup_{k}\|x_k\| + \|p_k\| < \infty\). We define a Lyapunov function candidate $V(p) = (p-p^\dagger)^{\top} M^{-\top} (p-p^\dagger)$ for the dynamical system \eqref{subeq:p_fin12}.  Note that \(V(\pEqFin{}) = 0\) and \(V(p) > 0\) for all \(p \neq \pEqFin{}\). Next, we show that \(\nabla V(p)^\top (e(x^\ast(p)) - p) < 0\) for every \(p \neq p^\dagger\).}
  Indeed, 
    \begin{align*}
        &\nabla V(p)^\top(e(x^\ast(p))-p) = 2(p-p^\dagger)^\top M^{-\top}\nabla h(x^\ast(p)) \\
        &= -2(x^\ast(p)-x^\ast(p^\dagger))\nabla h(x^\ast(p))          \\ 
        &\stackrel{\eqref{eq: PdaggerReformulate}}{=} -2(x^\ast(p)-x^\ast(p^\dagger))\left(\nabla h(x^\ast(p)) - \nabla h(x^\ast(p^\dagger))\right) \\ 
         &= -2(x^\ast(p)-x^\ast(p^\dagger))\left(\nabla h(x^\ast(p)) - \nabla h(x^\ast(p^\dagger))\right)  \\ 
         &< 0, \quad \forall \ p \neq p^\dagger, 
    \end{align*}
where the last equality follows from the strict convexity of \(h_i\) for each \(i \in \playerSet\), completing the proof.

\item[(b)] {{} Proposition \ref{prop: PropGlobalAggregativeGen}-(ii) follows by verifying conditions (R1) and (R2) of Proposition \ref{thm: ConvergenceFin}. Given that \(\sup_{k}\|x_k\| + \|p_k\| < \infty\), it suffices to verify (R1). This follows since condition (C1) in Lemma \ref{lem: SufficientConditionsPropMain} holds under Proposition \ref{prop: PropGlobalAggregative}-(ii) and Assumption \ref{assm: CostStructureAggregative}.  
}

First, we show that for \( i, j \in \playerSet \) with \( i \neq j \), it holds that \( \frac{\partial e_i(x^\ast(p))}{\partial p_j} > 0 \). Indeed,  
\begin{align*}
    \frac{\partial e_i(x^\ast(p))}{\partial p_j} &= \nabla^2 h_i(x_i^\ast(p))\frac{\partial x_i^{\ast}(p)}{\partial p_j} \\  
    &= \nabla^2 h_i(x_i^\ast(p))(-M^{-1})_{ij} > 0,  
\end{align*}
where the inequality follows from the strict convexity of \( h_i \) and the fact that \( (M^{-1})_{ij} < 0 \).  
Second, we show that condition (C1)-(i) in Lemma \ref{lem: SufficientConditionsPropMain} holds. 
First, we establish that \( e_i(x^\ast(0)) \geq 0 \) for every \( i \in \playerSet \). From \eqref{eq: Externality_AggGame}, we note that \( e_i(x^\ast(0)) = \nabla h_i(0) \) for every \( i \in \playerSet \). Therefore, it suffices to show that \( \nabla h_i(0) \geq 0 \) for all \( i \in \playerSet \). 
By Assumption \ref{assm: CostStructureAggregative}, \( \nabla h_i(\cdot) \) is strictly increasing, and for each \( i \in \playerSet \), there exists a unique \( y_i^\dagger \leq 0 \) such that \( \nabla h_i(y_i^\dagger) = 0 \). This implies that \( \nabla h_i(0) \geq 0 \), for every \(i\in \playerSet\). 
Next, we verify that \( p^\dagger \in \mathbb{R}_+^{|\playerSet|} \). From Proposition \ref{thm: UniquenessAggregativeGen}, \( p^\dagger = -My^\dagger \). Since \( M \) has non-negative entries and \( y^\dagger \in \mathbb{R}_-^{|\playerSet|} \), it follows that \( p^\dagger \in \mathbb{R}_+^{|\playerSet|} \).
Finally, we show the other condition in (C1)-(i) in Lemma \ref{lem: SufficientConditionsPropMain}, which requires that for any \( p \in \mathbb{R}_{+}^{|\playerSet|} \), there exists \( p' \in \mathbb{R}_{+}^{|\playerSet|} \) such that for every \( i \in \playerSet \), \( p'_i > p_i \) and \( e_i(x^\ast(p')) - p_i' \leq 0 \), for all \( i \in \playerSet \).
To show this, we define \( p^\epsilon = -(1+\epsilon)My^\dagger \) for every \( \epsilon > 0 \). Note that \(p^\epsilon \in \mathbb{R}_+^{|\playerSet|}\) and for any \(p\in \mathbb{R}^{|\playerSet|}_+\), we can select \(\epsilon > 0\) such that \(p^\epsilon_i > p_i\) for every \(i\in \playerSet\). Therefore, we show that for every \( \epsilon > 0\),  
\begin{align}\label{eq: e_diff}
e_i(x^\ast(p^\epsilon)) - p^\epsilon \leq 0, \quad \forall i \in \playerSet.
\end{align}
From \eqref{eq: Externality_AggGame}, we note that  \(
e_i(x^\ast(p^\epsilon)) - p^\epsilon = \nabla h_i(x^\ast_i(p^\epsilon)),
\) for every \(i\in \playerSet.\)
Therefore, to show \eqref{eq: e_diff}, it is sufficient to show that
\begin{align}\label{eq: h_diff}
\nabla h_i(x^\ast_i(p^\epsilon)) \leq 0, \quad \forall i \in \playerSet, \epsilon > 0.
\end{align}
Indeed, for every \(i\in \playerSet\) and \(\epsilon> 0\), 
\begin{align*}
0 &< (\nabla h_i(x_i^\ast({p}^\epsilon)) - \nabla h_i(y^\dagger_i))(x_i^\ast(p^\epsilon)-y^\dagger_i) \\ 
&= \nabla h_i(x_i^\ast(p^\epsilon))\epsilon y^\dagger_i,
\end{align*}
where we note that \( x^\ast(p^\epsilon) = (1+\epsilon)y^\dagger \). To conclude, 
\eqref{eq: h_diff} follows because \(y^\dagger_i \leq 0\) and \(\epsilon > 0\).
\end{itemize}

\section{Proofs of Results in Section \ref{sec:routing}}
\subsection{Proof of Proposition \ref{thm: FixedPointsRouting}.}\label{sssec: RoutingFixedPoint}
First, we show that \(\tilde{\textbf{P}}^\dagger\) is non-empty. This can be shown analogously to the proof of existence in Proposition \ref{prop: Alignment} by using the Schauder fixed-point theorem and the continuity of the function \(\tilde{w}^\ast(\cdot)\). We omit the details of this proof for the sake of brevity.

Next, we show that any \(\pEqFin{}\in \tilde{\textbf{P}}^\dagger\) aligns the Nash equilibrium with social optimality, i.e. \(\tilde{w}(\peq) = w^\dagger\). 
For any \(\pEqFin{}\in \tilde{\textbf{P}}^\dagger\), we have \(\optToll_a = \pertEdge_a(\optToll)\nabla \lossEdge_a(\pertEdge_a(\optToll))\) for every \(a\in \edge\). This implies, for every \(a\in \edge,\)
\begin{align}\label{eq: EqualTraffic}
\frac{\partial}{\partial \tilde{w}_a} \left(\tilde{w}_a(\tilde{p}^\dagger)\lossEdge_a(\tilde{w}_a(\tilde{p}^\dagger)) \right) = \lossEdge_a(\tilde{w}_a(\tilde{p}^\dagger)) + \tilde{p}^\dagger_a.
\end{align}

Note that for any arbitrary edge toll \(\edgeTolls\in \R^{|\edge|}\), \(\tilde{w}^\ast(\edgeTolls)\) is the unique solution to the following strictly convex optimization problem \cite{sandholm2010population}. 
\begin{equation}\label{eq: StochasticPerturbedEq}
\begin{aligned}
    &\min_{\tilde{w}\in\tilde{W}} &&\tilde{T}(\tilde{w}) = \sum_{a\in \edge} \int_{0}^{\flowEdge_a}\lossEdge_a(\tau)\,d\tau + \sum_{a\in \edge}\edgeTolls_a\flowEdge_a.
\end{aligned}
\end{equation}
Therefore, \(\tilde{w}^\ast(\tilde{p})\) is a Nash equilibrium if and only if 
\begin{align}\label{eq: VarEqnon-atomicTraffic}
\sum_{a\in \edge} \Big( \lossEdge_a\big(\tilde{w}_a(\tilde{p})+\tilde{p}_a\big)\big(\tilde{w}_a-\tilde{w}_a(\tilde{p})\big) \Big) \geq 0,  ~\forall \ \tilde{w} \in \tilde{W}. 
\end{align}
Combining \eqref{eq: EqualTraffic} and \eqref{eq: VarEqnon-atomicTraffic}, we conclude that for every \(\tilde{w} \in \tilde{W},\)
\begin{align}\label{eq: VarEqnon-atomicSocCostTraffic}
\sum_{a\in\edge}\frac{\partial}{\partial \tilde{w}_a}(\tilde{w}_a(\tilde{p}^\dagger)\lossEdge_a(\tilde{w}_a(\tilde{p}^\dagger))) (\tilde{w}_a-\tilde{w}_a^\ast(\tilde{p}^\dagger)) \geq 0. 
\end{align}
Further, from the first-order conditions of optimality for the social cost function, we know that \(\tilde{x}^\dagger\) is socially optimal if and only if, for every \(\tilde{x}\in \strategySetPop,\)
\begin{align}\label{eq: SocCostTrafficDer}
    \sum_{i\in \pop}\sum_{j\in \routes_i} \frac{\partial \Phi(\tilde{x}^\dagger)}{\partial \tilde{x}_i^j} (\tilde{x}_i^j-\tilde{x}_i^{\dagger j}) \geq 0. 
\end{align}
Using Lemma \ref{lem: DerPotTrafficExpression} in Appendix \ref{sec: AppProofs}, we can equivalently write \eqref{eq: SocCostTrafficDer} in terms of edge flows as follows 
\begin{align}\label{eq: SocCostFiniteVInon-atomicTraffic}
\sum_{a\in\edge}\frac{\partial}{\partial \tilde{w}_a}(\tilde{w}_a^\dagger\lossEdge_a(\tilde{w}_a^\dagger)) (\tilde{w}_a-\tilde{w}_a^\dagger) \geq 0  \quad \forall \ \tilde{w} \in \tilde{W}, 
\end{align}
where \(w^\dagger\) is the edge flow corresponding to the route flow \(x^\dagger\). 
Comparing \eqref{eq: VarEqnon-atomicSocCostTraffic} with \eqref{eq: SocCostFiniteVInon-atomicTraffic}, we note that \(\tilde{w}^\ast(p^\dagger)\) is the minimizer of social cost function \(\socCostPop\). Therefore, \(\tilde{w}^\ast(\tilde{p}^\dagger) = \tilde{w}^\dagger\). 

The proof that \(\tilde{\textbf{P}}^\dagger\) is a singleton follows by contradiction, which is analogous to that in Proposition \ref{prop: Alignment}. We omit the details for the sake of brevity.

\subsection{Proof of Proposition \ref{thm: ConvergenceRouting}}\label{sssec: ConvergenceRouting}
{{} The proof follows by verifying the requirements (R1)-(R2) in Proposition \ref{thm: ConvergenceFin}.  
Requirement (R2) holds since the strategy space is a compact set.  
It suffices to show that requirement (R1) holds.}  
Towards this goal, we define a Lyapunov function candidate  \(
V(\edgeTolls) = (\edgeTolls - \optToll)^\top \Delta (\edgeTolls - \optToll)\) 
for the dynamical system \eqref{eq: DynTraffic}, where \(\Delta \in \mathbb{R}^{|\edge| \times |\edge|}\)  
is a diagonal matrix defined in \eqref{eq: DeltaMat}.  
Due to the strict monotonicity and convexity of \(\lossEdge_a(\cdot)\),  
it follows that \(\Delta_{a,a} > 0\) for every \(a \in \edge\). Consequently, the Lyapunov function candidate is positive definite.

We show that there exists a positive scalar \(r\) such that for any \(\edgeTolls \in \mathcal{B}_r(\optToll)\), the following holds:
\begin{align}\label{eq: CondCheckRouting}
\sum_{a \in \edge} \nabla_{\edgeTolls_a} V(\edgeTolls)^\top \left(\pertEdge_a(\edgeTolls) \nabla \lossEdge_a(\pertEdge_a(\edgeTolls)) - \edgeTolls_a \right) < -2V(\edgeTolls).
\end{align}
Indeed, we note that
\begin{align*}
    &\sum_{a \in \edge} \nabla_{\edgeTolls_a} V(\edgeTolls) \left(\pertEdge_a(\edgeTolls) \nabla \lossEdge_a(\pertEdge_a(\edgeTolls)) - \edgeTolls_a \right) \\
    &= 2 \sum_{a \in \edge} \Delta_{a,a} (\edgeTolls_a - \optToll_a) \left(\pertEdge_a(\edgeTolls) \nabla \lossEdge_a(\pertEdge_a(\edgeTolls)) - \edgeTolls_a \right) \\
    &= 2 \sum_{a \in \edge} \Delta_{a,a} (\edgeTolls_a - \optToll_a) \left( \pertEdge_a(\edgeTolls) \nabla \lossEdge_a(\pertEdge_a(\edgeTolls)) - \optToll_a + \optToll_a - \edgeTolls_a \right) \\
    &= -2 V(\edgeTolls) + 2 \sum_{a \in \edge} \Delta_{a,a} (\edgeTolls_a - \optToll_a) \left( \phi_a(\edgeTolls) - \phi_a(\optToll) \right),
\end{align*}
where for every \(a \in \edge\), \(\phi_a(\edgeTolls) := \pertEdge_a(\edgeTolls) \nabla \lossEdge_a(\pertEdge_a(\edgeTolls))\). Thus, to show local convergence, it suffices to show that there exists \(r > 0\) such that
\begin{align*}
    \sum_{a \in \edge} \Delta_{a,a} (\edgeTolls_a - \optToll_a) \left( \phi_a(\edgeTolls) - \phi_a(\optToll) \right) \leq 0, \quad \forall \tilde{p} \in \mathcal{B}_r(\tilde{p}^\dagger).
\end{align*}

To show this, we note that due to condition \eqref{eq: degenerate}, the function \(\phi\) is differentiable in a neighborhood of \(\tilde{p}^\dagger\) (cf. \cite[Chapter 4]{yang2005mathematical}). Consequently, using Lemma \ref{lem: MonotonicityGradient} in Appendix \ref{sec: AppProofs}, it is sufficient to show that 
{{} \begin{align}\label{eq: PhiCond}
\sum_{a,a' \in \tilde{\mathcal{E}}} z_a \Delta_{a,a} \frac{\partial \phi_a(\tilde{p}^\dagger)}{\partial \tilde{p}_{a'}} z_{a'} \leq 0, \quad \forall \ z \in \mathbb{R}^{|\tilde{\mathcal{E}}|}.
\end{align}}
Indeed, by the design of \(\Delta\), it holds that 
\begin{align}\label{eq: Delta_w_Phi}
\Delta_{a,a} \frac{\partial \phi_a(\tilde{p}^\dagger)}{\partial \tilde{p}_{a'}} = \frac{\partial \tilde{w}^\ast_a(\tilde{p}^\dagger)}{\partial \tilde{p}_{a'}}, \quad \forall \ a, a' \in \tilde{\mathcal{E}}.
\end{align}
Furthermore, Lemma \ref{lem: MonotonicityEdgeFlow} and Lemma \ref{lem: MonotonicityGradient} in Appendix \ref{sec: AppProofs} guarantee that 
{{} \begin{align}\label{eq: w_Phi}
\sum_{a,a' \in \tilde{\mathcal{E}}} z_a \frac{\partial \tilde{w}^\ast_a(\tilde{p}^\dagger)}{\partial \tilde{p}_{a'}} z_{a'} \leq 0, \quad \forall \ z \in \mathbb{R}^{|\tilde{\mathcal{E}}|}.
\end{align}}
The proof concludes by noting that \eqref{eq: Delta_w_Phi} and \eqref{eq: w_Phi} imply \eqref{eq: PhiCond}.

\section{Auxiliary Results}\label{sec: AppProofs}
{{} \begin{lemma}\label{lem: SufficientConditionsPropMain}
Requirement (R1) of Proposition \ref{thm: ConvergenceFin} is satisfied if either one of the following conditions holds:
\begin{itemize}
    \item[(C1)] \(\frac{\partial \mdFin_i(\xEqFin{}(\incentiveFin))}{\partial \incentiveFin_j} > 0\) for all \(\incentiveFin \in \R^{\actDimFin}\) and all \(i \neq j\), and at least one of the following conditions holds:
    \begin{itemize}
        \item[\textit{(i)}] \(e_i(x^\ast(0)) \geq 0\) for every \(i \in \playerSet\), \(p^\dagger\in \mathbb{R}_+^{|\playerSet|},\) and for any \(p \in \R_{+}^{|\playerSet|}\), there exists \(p'\in \R_{+}^{|\playerSet|}\) such that \(p'_i>p_i\) and 
    \(
    e_i(x^\ast(p')) - p_i' \leq 0~ \text{for every } i \in \playerSet.
    \)
    Moreover, \(\xk{0} \in X, \pk{0} \in \mathbb{R}^{|\playerSet|}_+\).
    \item[\textit{(ii)}] \(e_i(x^\ast(0)) \leq 0\) for every \(i \in \playerSet\), \(p^\dagger\in \mathbb{R}_-^{|\playerSet|},\) and for any \(p \in \R_{-}^{|\playerSet|}\), there exists \(p'\in \R_{-}^{|\playerSet|}\) such that \(p'_i<p_i\) and 
    \(
    e_i(x^\ast(p')) - p_i' \geq 0~ \text{for every } i \in \playerSet.
    \)
    Moreover, \(\xk{0} \in X, \pk{0} \in \mathbb{R}^{|\playerSet|}_-\).
    \end{itemize}
    \item[(C2)] There exists a set \(\textsf{dom}(V) \subset \mathbb{R}^{|\playerSet|}\) and a continuously differentiable function \(V: \textsf{dom}(V) \rightarrow \mathbb{R}_+\) such that \(V(\pEqFin{}) = 0\) and \(V(p) > 0\) for all \(p \neq \pEqFin{}\). Moreover, for every \(p \neq p^\dagger,\) \(
        \nabla V(p)^\top \lr{e(\xEqFin{}(p)) - p} < 0.\)
\end{itemize}
\end{lemma}
\begin{proof}
    Conditions (C1) and (C2) above are based on results from non-linear dynamical systems which ensure convergence of \eqref{subeq:p_fin11}. In particular, (C1)-(i) (resp. (C1)-(ii)) builds on cooperative dynamical systems theory \cite{hirsch1985systems}, which ensures that  \(\mathbb{R}_+^{|\playerSet|}\) (resp. \(\mathbb{R}_-^{|\playerSet|}\)) is positively invariant for \eqref{subeq:p_fin11}  and \(p^\dagger \in \mathbb{R}_+^{|\playerSet|}\) (resp. \(p^\dagger \in \mathbb{R}_-^{|\playerSet|}\)) is asymptotically stable. On the other hand, condition (C2) ensures the existence of a Lyapunov function that is strictly positive everywhere except at \(p^\dagger\) and decreases along any trajectory of \eqref{subeq:p_fin11} (cf. \cite{sastry2013nonlinear}).
\end{proof}
}

\begin{lemma}\label{lem: ExternalityExpression}
For every \(i\in \pop, j\in \routes_i,\)  
\(
    \exterPop_{i}^j(\flowRoute) = \sum_{a\in j} \flowEdge_a \nabla \lossEdge_a(\flowEdge_a).
\)
\end{lemma}
\begin{proof}
Using \eqref{eq: ExternalityRoutingRoute}, we note that 
\begin{align*}
    &\exterPop_{i}^j(\flowRoute) = \sum_{i'\in\odPair}\sum_{j'\in\routes_i}\flowRoute_{i'}^{j'}\frac{\partial \lossRoute_{i'}^{j'}(\flowRoute)}{\partial \flowRoute_{i}^j}\\
    &\stackrel{(a)}{=} \sum_{i'\in\odPair}\sum_{j'\in\routes_i}\flowRoute_{i'}^{j'}\sum_{a\in \edge} \mathbbm{1}(a\in j') \nabla \lossEdge_a(\flowEdge_a) \frac{\partial \flowEdge_a}{\partial \flowRoute_{i}^j}\\
    &\stackrel{(b)}{=} \sum_{i'\in\odPair}\sum_{j'\in\routes_i}\flowRoute_{i'}^{j'}\sum_{a\in \edge} \mathbbm{1}(a\in j') \nabla \lossEdge_a(\flowEdge_a) \mathbbm{1}(a\in j)\\
    &\stackrel{(c)}{=} \sum_{a\in j} \nabla \lossEdge_a(\flowEdge_a) \flowEdge_a,
\end{align*}
where \((a)\) follows by expanding out the expression of route costs in terms of edge costs and using the chain rule, \((b)\) follows by the definition of edge flows, and \((c)\) follows by changing the order of summations and using the definition of edge flows. This completes the proof.
\end{proof}

\begin{lemma}\label{lem: DerPotTrafficExpression}
\(x^\dagger\) that satisfies \eqref{eq: SocCostTrafficDer} if and only if \eqref{eq: SocCostFiniteVInon-atomicTraffic}. 
\end{lemma}
\begin{proof}
    First, we show that \(
        \Phi(\tilde{x}) = \sum_{a\in \edge}\tilde{w}_a\lossEdge_a(\tilde{w}_a).\)
    \begin{align*}
        \Phi(\tilde{x}) &\stackrel{\eqref{eq: SocCostTraffic}}{=} \sum_{i\in\odPair}\sum_{j\in\routes_i}\flowRoute_{i}^j\lossRoute_{i}^j(\flowRoute)  =  \sum_{i\in\odPair}\sum_{j\in\routes_i}\flowRoute_{i}^{j}\sum_{a\in \edge} \mathbbm{1}(a\in j) \lossEdge_{a}(\flowEdge_a) 
    \\
    & = \sum_{a\in \edge} \lossEdge_a(\flowEdge_a)\sum_{i\in\odPair}\sum_{j\in\routes_i}\flowRoute_{i}^{j}\mathbbm{1}(a\in j) = \sum_{a\in \edge} \lossEdge_a(\flowEdge_a)\flowEdge_a.
    \end{align*}
    Next, observe that  
    \begin{align*}
        &\sum_{i\in \pop}\sum_{j\in \routes_i} \frac{\partial \Phi(\tilde{x}^\dagger)}{\partial \tilde{x}_i^j} (\tilde{x}_i^j-\tilde{x}_i^{\dagger j}) 
        \\
        &= \sum_{i\in \pop}\sum_{j\in \routes_i} \sum_{a\in\edge}\frac{\partial}{\partial \tilde{x}_i^j}\left(\flowEdge_a\lossEdge_a(\flowEdge_a) \right) (\tilde{x}_i^j-\tilde{x}_i^{\dagger j})\\
        &= \sum_{i\in \pop}\sum_{j\in \routes_i} \sum_{a\in\edge}\frac{\partial}{\partial \tilde{w}_a}\left(\flowEdge_a\lossEdge_a(\flowEdge_a) \right)\mathbbm{1}(a\in j) (\tilde{x}_i^j-\tilde{x}_i^{\dagger j})
        \\&= \sum_{a\in\edge}\frac{\partial}{\partial \tilde{w}_a}\left(\flowEdge_a\lossEdge_a(\flowEdge_a) \right) (\tilde{w}_a-\tilde{w}_a^{\dagger}). 
    \end{align*}
    This concludes the proof. 
\end{proof}

\begin{lemma}\label{lem: MonotonicityEdgeFlow} Following inequality holds: 
\begin{align}\label{eq: EdgeFlowMonotonicity}
    \sum_{a \in \edge} (\edgeTolls_a - \edgeTolls'_a) \left( \pertEdge_a(\edgeTolls_a) - \pertEdge_a(\edgeTolls'_a) \right) \leq 0, \quad \forall \ \edgeTolls, \edgeTolls' \in \R^{|\edge|}.
\end{align}
\end{lemma}
\begin{proof}
To prove this result, we first show that 
\begin{align}\label{eq: RouteFlowMonotonicity}
\sum_{i\in\odPair}\sum_{j\in\routes_i} (\routeTolls_{i}^j-\routeTolls_{i}^{'j})(\tilde{x}_i^{\ast j}(\routeTolls)-\tilde{x}_i^{\ast j}(\routeTolls')) \leq 0,
    \end{align}
    where \(\routeTolls\) and \(\routeTolls'\) are the route tolls associated with edge tolls \(\edgeTolls\) and \(\edgeTolls'\), respectively, through \eqref{eq: RelEdgeRouteToll}. 
Let the feasible set of route flows in the optimization problem \eqref{eq: StochasticPerturbedEq} be denoted by \(\mathcal{F}\). Using the first-order conditions of optimality for the strictly convex optimization problem \eqref{eq: StochasticPerturbedEq}, we obtain:
\begin{equation}\label{eq: FOCPertEq}
\begin{aligned}
    \sum_{i\in\odPair}\sum_{j\in\routes_i}&\left( \loss_{i}^j(\pertEq(\routeTolls),\routeTolls) \right)\cdot \left(\tilde{y}_{i}^j-\tilde{x}_i^{\ast j}(\routeTolls)\right) \geq 0, \forall \quad \tilde{y}\in \mathcal{F},
\end{aligned}
\end{equation}
where \(\routeTolls\) is the route toll associated with edge toll \(\edgeTolls\). 
Rewriting \eqref{eq: FOCPertEq} for edge tolls \(\edgeTolls'\) we obtain 
\begin{equation}\label{eq: FOCPertEqSim}
\begin{aligned}
    \sum_{i\in\odPair}\sum_{j\in\routes_i}&\left( \loss_{i}^j(\pertEq(\routeTolls'),\routeTolls') \right)\cdot \left(\tilde{y}_{i}^{'j}-\tilde{x}_i^{\ast j}(\routeTolls')\right) \geq 0, \forall \quad \tilde{y}'\in \mathcal{F},
\end{aligned}
\end{equation}
where \(\routeTolls'\) is the route toll associated with edge toll \(\edgeTolls'\).

Next, we prove \eqref{eq: RouteFlowMonotonicity}. Note that 
\begin{align*}
    & \sum_{i\in\odPair}\sum_{j\in\routes_i}(\routeTolls_i^j-\routeTolls_{i}^{'j})(\tilde{x}_i^{\ast j}(\routeTolls)-\tilde{x}_i^{\ast j}(\routeTolls'))\notag \\
    &\stackrel{(a)}{\leq} \sum_{i\in\odPair}\sum_{j\in\routes_i}(\lossRoute_{i}^j(\pertEq(\routeTolls'))-\lossRoute_{i}^j(\pertEq(\routeTolls)))(\tilde{x}_i^{\ast j}(\routeTolls)-\tilde{x}_i^{\ast j}(\routeTolls'))  \notag
    \\
    &\stackrel{(b)}{=} \sum_{i\in\odPair}\sum_{j\in\routes_i}(\tilde{x}_i^{\ast j}(\routeTolls)-\tilde{x}_i^{\ast j}(\routeTolls'))\cdot \notag\\ 
    &\hspace{2cm}\cdot\sum_{a\in \edge} (\lossEdge_{a}(\pertEdge_a(\edgeTolls'))-\lossEdge_{a}(\pertEdge_a(\edgeTolls)))\mathbbm{1}(a\in j) \notag \\ 
    &\stackrel{(c)}{=} \sum_{a\in \edge}(\lossEdge_{a}(\pertEdge_a(\edgeTolls'))-\lossEdge_{a}(\pertEdge_a(\edgeTolls)))\cdot\notag \\ &\hspace{2cm}\cdot
    \sum_{i\in\odPair}\sum_{j\in\routes_i}(\tilde{x}_i^{\ast j}(\routeTolls)-\tilde{x}_i^{\ast j}(\routeTolls'))\mathbbm{1}(a\in j)\notag \\
    &\stackrel{(d)}{=} \sum_{a\in \edge}(\lossEdge_{a}(\pertEdge_a(\edgeTolls'))-\lossEdge_{a}(\pertEdge_a(\edgeTolls)))(\pertEdge_{a}(\edgeTolls)-\pertEdge_{a}(\edgeTolls')) \stackrel{(e)}{\leq } 0,
\end{align*}
where we obtain \((a)\) by adding \eqref{eq: FOCPertEq}, evaluated at \(\tilde{y} = \pertEq(\routeTolls')\), and \eqref{eq: FOCPertEqSim}, evaluated at \(\tilde{y}' = \pertEq(\routeTolls)\), \((b)\) holds by the definition of the route loss function, \((c)\) holds by interchange of summation, \((d)\) holds by the definition of edge flows, and \((e)\) holds due to the monotonicity of edge latency functions. This proves \eqref{eq: RouteFlowMonotonicity}.

Finally, we prove \eqref{eq: EdgeFlowMonotonicity}. Note that 
\begin{align*}
    &\sum_{a\in\edge} (\edgeTolls_a-\edgeTolls_{a'})(\pertEdge_a(\edgeTolls)-\pertEdge_a(\edgeTolls')) \\ 
    &\stackrel{(a)}{=}\sum_{a\in\edge} (\edgeTolls_a-\edgeTolls_{a'})\sum_{i\in\odPair}\sum_{j\in\routes_i}(\tilde{x}_i^{\ast j}(\routeTolls)-\tilde{x}_i^{\ast j}(\routeTolls'))\mathbbm{1}(a\in j)\\
    &\stackrel{(b)}{=} \sum_{i\in\odPair}\sum_{j\in\routes_i}(\tilde{x}_i^{\ast j}(\routeTolls)-\tilde{x}_i^{\ast j}(\routeTolls'))\sum_{a\in\edge} (\edgeTolls_a-\edgeTolls_{a'})\mathbbm{1}(a\in j) \\ 
    &\stackrel{(c)}{=} \sum_{i\in\odPair}\sum_{j\in\routes_i}(\tilde{x}_i^{\ast j}(\routeTolls)-\tilde{x}_i^{\ast j}(\routeTolls'))(\routeTolls_i^j-\routeTolls_i^{'j})\stackrel{(d)}{\leq} 0,
\end{align*}
where \((a)\) holds due to the definition of edge flows, \((b)\) holds due to interchange of summation, \((c)\) holds due to the definition of route tolls, and \((d)\) holds due to \eqref{eq: RouteFlowMonotonicity}. This concludes the proof.
\end{proof}

\begin{lemma}[\cite{facchinei2007finite}]\label{lem: MonotonicityGradient}
For any fixed \(p'\) and continuously differentiable function \(\phi:\R^{\edge} \rightarrow \R^{\edge}\), the condition
\[
\langle \phi(p) - \phi(p'), p - p' \rangle \leq 0 \quad \forall \ p \in \mathcal{B}_r(p')
\]
for some \(r > 0\), holds if and only if
{{} \[
\sum_{i,j \in |\tilde{\mathcal{E}}|} z_i z_j \frac{\partial \phi_i(p')}{\partial p_j} \leq 0, \quad \forall \ z \in \R^{|\edge|}.
\]} 
\end{lemma}
\bibliography{refs}
\bibliographystyle{apalike}

\appendix

\end{document}